\documentclass{article}
\usepackage{graphicx}
\usepackage{amsfonts}
\usepackage{amsmath}
\usepackage{amssymb}
\usepackage{url}
\usepackage{fancyhdr}
\usepackage{indentfirst}
\usepackage{enumerate}

\usepackage{caption} 

\usepackage{dsfont}
\usepackage[colorlinks=true,citecolor=blue]{hyperref}
\usepackage{amsthm}
\usepackage{color}
\usepackage{natbib}
\usepackage{comment}

\def\d{\mathrm{d}}
\def\laweq{\buildrel \mathrm{d} \over =}

\newcommand{\X}{\mathcal {X}}

\newcommand{\VaR}{\mathrm{VaR}}

\newcommand{\ES}{\mathrm{ES}}
\newcommand{\ex}{\mathrm{ex}}
\newcommand{\DQex}{\mathrm{DQ}^\mathrm{ex}_\alpha}
\newcommand{\DQVaR}{\mathrm{DQ}^\mathrm{VaR}_\alpha}
\newcommand{\DQES}{\mathrm{DQ}^\mathrm{ES}_\alpha}

\newcommand{\E}{\mathbb{E}}

\newcommand{\R}{\mathbb{R}}

\newcommand{\p}{\mathbb{P}}

\newcommand{\id}{\mathds{1}}

\renewcommand{\ge}{\geqslant}
\renewcommand{\le}{\leqslant}
\renewcommand{\geq}{\geqslant}
\renewcommand{\leq}{\leqslant}
\renewcommand{\epsilon}{\varepsilon}
\DeclareMathOperator*{\argmax}{arg\,max}
\DeclareMathOperator*{\argmin}{arg\,min}
\newcommand{\esssup}{\mathrm{ess\mbox{-}sup}}

\theoremstyle{plain}
\newtheorem{theorem}{Theorem}

\newtheorem{proposition}{Proposition}
\theoremstyle{definition}
\newtheorem{definition}{Definition}
\newtheorem{example}{Example}

\theoremstyle{remark}
\newtheorem{remark}{Remark}

\newcommand{\cet}{\begin{center}}
\newcommand{\ecet}{\end{center}}

\begin{document}

\title{Diversification quotient based on expectiles}

\author{
Xia Han\thanks{School of Mathematical Sciences and LPMC, Nankai University, China.
			\texttt{xiahan@nankai.edu.cn}}
\and Liyuan Lin\thanks{Department of Econometrics and Business Statistics, Monash University,  Australia. \texttt{liyuan.lin@monash.edu}}
\and 	Hao Wang\thanks{School of Mathematical Sciences, Nankai University, China. \texttt{hao.wang@mail.nankai.edu.cn}}	\and
Ruodu Wang\thanks{Department of Statistics and Actuarial Science, University of Waterloo, Canada.   \texttt{wang@uwaterloo.ca}}
}
\date{\today}
\maketitle

\begin{abstract}  

A diversification quotient (DQ) quantifies diversification in stochastic portfolio models based on a family of risk measures.
We study DQ  based on expectiles, offering a useful alternative to conventional risk measures such as Value-at-Risk (VaR) and Expected Shortfall (ES). The expectile-based DQ admits simple formulas and has a natural connection to the Omega ratio.
Moreover, the expectile-based DQ is not affected by small-sample issues faced by VaR-based or ES-based DQ due to the scarcity of tail data.
 The expectile-based  DQ  exhibits pseudo-convexity in portfolio weights, allowing  gradient descent algorithms for portfolio selection. We show that the corresponding optimization problem can be efficiently solved using linear programming techniques in real-data applications. Explicit formulas for  DQ based on expectiles are also derived for  elliptical and multivariate regularly varying  distribution models.  Our findings enhance the understanding of the DQ's role in financial risk management and highlight its potential to improve portfolio construction strategies.

\end{abstract}

\section{Introduction}

Various approaches have been developed in the literature to quantify the degree of diversification, which is useful in portfolio selection. Classic diversification indices, such as the diversification ratio (DR) and the diversification benefit (DB), evaluate changes in capital requirements caused by pooling assets. We refer the reader to \cite{T07}, \cite{CC08}, \cite{EWW15}, \cite{EFK09}, and \cite{MFE15} for works related to DR and DB based on Value-at-Risk (VaR), Expected Shortfall (ES), or other commonly used risk measures.

Recently, \cite{HLW24} proposed a new diversification index called the diversification quotient (DQ), which quantifies the improvement in the risk-level parameter by the diversification effect. 
DQ can be characterized through six axioms: non-negativity, location invariance, scale invariance, rationality, normalization, and continuity. This characterization distinguishes DQ from other indices such as DR and DB, which do not have an axiomatic foundation. In financial risk management, VaR and ES are widely used, particularly within regulatory frameworks such as Basel III/IV and Solvency II.
DQs based on VaR and ES has many appealing features, including distinguishing heavy tails and common shocks, convenient formulas for computation, and efficiency in portfolio optimization, as shown in \cite{HLW23,HLW24}. Furthermore, they also study additional properties of DQs based on VaR and ES under popular models in risk management.

  


Besides VaR and ES, expectiles are arguably the third most important class of risk measures in the recent literature. 
Expectiles were introduced by \cite{NP87} as an asymmetric generalization of the mean in regression analysis.
\cite{BKMR14} studied expectiles in the context of risk measures, and showed that they are coherent risk measures (\cite{ADEH99}). 
An important feature of expectiles is that they are    elicitable. Elicitability is the property of a statistical functional to be defined as the minimizer of a suitable expected loss, a feature that is useful in statistical forecast and machine learning; see  \cite{G11} and \cite{FK21}. \cite{Z16} uniquely identified expectiles as the only coherent risk measures that are also elicitable. For further developments on expectiles,  see \cite{BD17}    for using expectiles in risk management,  \cite{MY15} and \cite{HCM24} for  asymptotic properties of risk concentration, \cite{CW16} and \cite{XLMZ23} for  their applications in insurance, and  \cite{BMWW24} for an axiomatization of expectiles in decision theory, with a connection the disappointment theory of \cite{G91}.

Because expectiles are an important   class of risk measures and DQ can be constructed based on any monotonic parametric class of risk measures, this paper conducts a systemic study on DQ based on expectiles, and in particular on its financial implications. This will provide a larger toolbox for risk management beyond DQs based on VaR and ES.

The first question we need to address is what DQ based on expectiles offers, different from DQs based on VaR and ES. 
After all,  DQs  based on VaR and ES enjoy many useful features, 
including explicit formulas, natural interpretations of dependence, and efficient optimization methods (\cite{HLW23,HLW24}). Indeed, DQ based on expectiles also shares most of these nice features. 
But more importantly, there are additional advantages of DQ based on expectiles that are not shared by DQ based on VaR or ES, which we explain below. 
\begin{enumerate}[(i)]
\item 
DQ based on expectiles connects to the well-established financial concept of the Omega ratio (see e.g.,  \cite{KS02}, \cite{KZCR14a,KZCR14b}, \cite{GMOS16} and \cite{SM17}), where the reference threshold in the Omega ratio is chosen as the sum of the expectiles of the individual risks (Theorem \ref{th:var}).
There are, however, two main differences: First, the Omega ratio does not reflect diversification, and it only evaluates based on the total risk; in contrast, the DQ is based on a balance between the total risk and the marginal risks, thus reflecting diversification.
Second, when an Omega ratio is applied in portfolio selection, the optimized portfolio is highly sensitive to the exogenously chosen threshold for the Omega ratio, but for DQ, the implicit threshold is endogenously computed from the expectiles. 
Compared to DQs based on VaR and ES, which focus solely on the tail risk, typically representing losses, DQ based on expectiles balances gains and losses.
An illustration of this advantage of DQ based on expectiles is provided in Example \ref{ex:stable} in Section \ref{sec:3}. 
    Figure \ref{sec_40_Omega} and Table \ref{tab_40_Omega} compare the performance of portfolios constructed from maximizing  Omega ratios with different thresholds and minimizing DQ based on expectile.
 As shown in Figure \ref{sec_40_Omega} and Table \ref{tab_40_Omega}, volatility varies significantly across different choices of thresholds. The strategy based on 
$\mathrm{DQ}_{0.1}^{\mathrm{ex}}$
  achieves the highest annual return and lowest annual volatility among all evaluated strategies, demonstrating that a strong diversification empirically has good return and volatility. For the case of 
$\alpha=0.05$, the DQ strategy performs also reasonably well.


    \begin{figure}[t]
\caption{Wealth processes for portfolios   maximizing  Omgea ratios with thresholds 0, $t_0$, $1.5t_0$, and $2 t_0$, where $t_0$ is the monthly return for the equally weighted portfolio that  minimizing DQ based on expectiles ($\DQex$) with parameter $\alpha =0.05$ and $\alpha =0.1$. The dataset has 40 stocks from 2014 to 2023, described in Section \ref{sec:62}. Portfolios are rebalanced each month.}\label{sec_40_Omega} 
\centering
           \includegraphics[height=5cm]{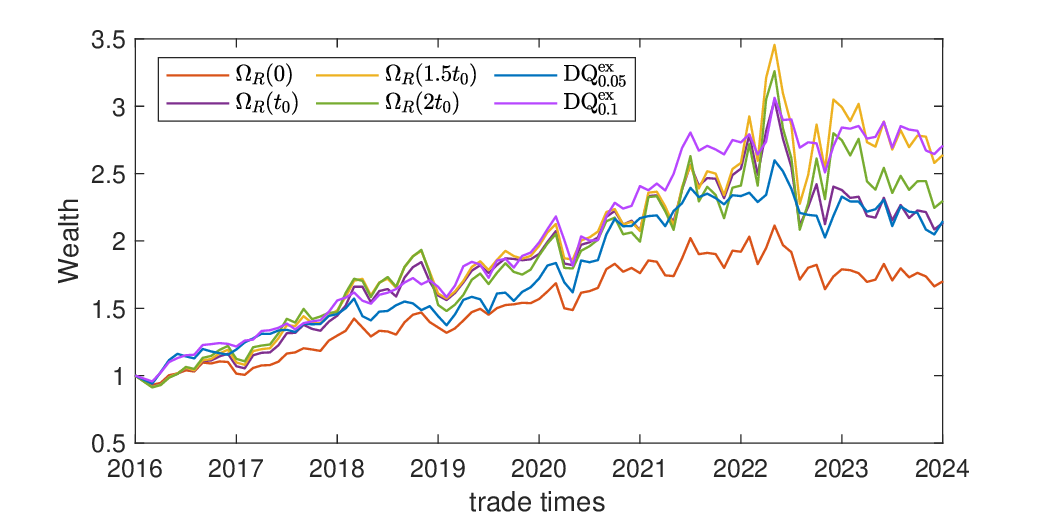}
\end{figure} 
  \begin{table}[t]  
\def\arraystretch{1}
  \begin{center} 
    \caption{ Summary statistics,  including the annualized return (AR), the annualized volatility (AV),  the Sharpe ratio (SR)  for   different portfolio strategies in Figure \ref{sec_40_Omega}.} \label{tab_40_Omega} 
  \begin{tabular}{c|cccccc} 
    $\%$ &${\rm DQ}^{\ex}_{0.05}$&${\rm DQ}^{\ex}_{0.1}$ &$\Omega_R(0)$ &$\Omega_R(t_0)$ & $\Omega_R(1.5t_0)$ &$\Omega_R(2t_0)$ \\    \hline
AR & 10.19 &\bf{13.10}  &  7.20 & 10.07& {12.67} & 10.91 \\
AV & {15.64}  & \bf{15.12} & 15.91 &  20.57&24.36&26.48 \\
SR  & {51.51} &  \bf{72.54}  & {31.93} &38.61 & {43.27}&  {33.13} 
\\
 \hline \hline 
    \end{tabular}
    \end{center}
    \end{table}

\item 
DQ based on expectiles can handle the issue of computation based on small sample, which is a practical challenge for VaR and ES, as well as their corresponding DQ, due to the lack of tail data.\footnote{For instance, if the level $\alpha$ in DQ based on VaR or ES is smaller than $1/N$, where $N$ is the sample size of empirical data, then the empirical estimate of DQ is always $0$;  see Remark \ref{rem:small-sample} in Section \ref{sec:4}.}
An illustration of a portfolio vector
from an equicorrelated normal distribution  is provided in Figure \ref{em_re}, where we can see that the sample estimate of DQ based on expectiles increases when the correlation coefficient increases, whereas the sample estimates of DQs based on VaR and ES stay  at $0$ regardless of the dependence structure. 
This shows that DQ based on expectile does not suffer from the small-sample issue as DQ based on VaR or ES in the context of empirical estimation.

\begin{figure}[t]
\caption{Empirical VS Real $\DQVaR$, $\DQES$ and $\DQex$ for $\mathbf X \sim \mathrm{N}(0, \Sigma)$ and $\alpha=0.02$ where $\Sigma=(\sigma)_{ij}$ with $\sigma_{ii}=1$ and $\sigma_{ij}=r$ for $i\ne j \in [5]$; the empirical value is the mean of 1000 sample estimates which is calculated based on  49 simulated data of $\mathbf X$.}\label{em_re}
      \hspace{-1.5cm}     \includegraphics[height=4.3cm]{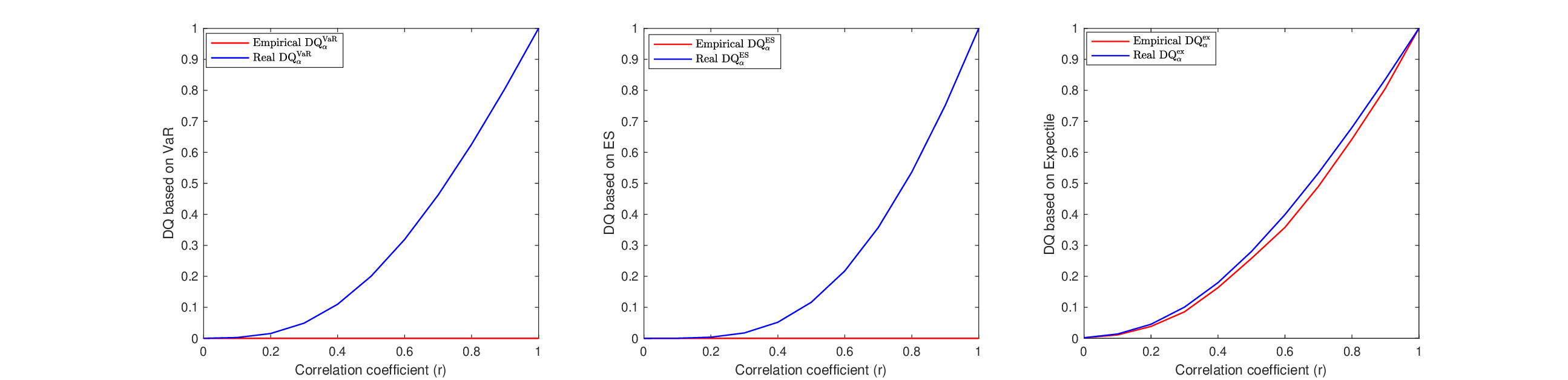}
\end{figure}



\item 
 DQ based on expectiles exhibits pseudo-convexity with respect to portfolio weights (Theorem \ref{prop:pseudo}), which is a property  not shared with DQs based on  VaR.\footnote{Whether DQ based ES is pseudo-convex is not clear and we do not have a proof; it was shown by \cite{HLW24} to satisfy a weaker property, which is quasi-convexity.} This pseudo-convexity allows for gridient descent algorithms, and simplifies the optimization process in portfolio selection. 
Similar to DQ based on VaR or ES, the portfolio optimization based on DQ with expectiles can be efficiently solved through linear programming, corresponding to maximizing the Omega ratio, or convex programming, analogous to the mean-variance optimization framework (Proposition \ref{thm:2}). Table \ref{table:time} reports the programming running time used in solve the portfolio optimization problem in Figure \ref{sec_40}. As we can see, DQ based on expectile shows comparable computional efficiency to other methods.

\begin{table}[!htbp]  
\def\arraystretch{1}
  \begin{center} 
    \caption{Computation times   for different portfolio strategies from   Jan 2016 - Dec 2023 with $\alpha=0.05$.}  \label{table:time}
 \begin{tabular}{c|cccc} 
     &${\rm DQ}^{\ex}_\alpha$ &${\rm DQ}^{\VaR}_\alpha$ &${\rm DQ}^{\ES}_\alpha$ & $\Omega_R(t_0)$  \\    \hline
Time (s) &8.47& 5.63   &    58.55   & 8.65  
\\
 \hline  \hline
\end{tabular}

    \end{center}
    \end{table}

\item Under a mild condition, the value  DQ based on expectiles is the unique solution of the adjust factor $c$ (Proposition \ref{prop:2}) for the risk level $\alpha$ in
$$\ex_{c\alpha} \left(\sum_{i=1}^n X_i\right) = \sum_{i=1}^n \ex_{\alpha}(X_i),$$
where $\ex_\alpha$ is the $\alpha$-expectile (defined in Section \ref{sec:2}) with $\alpha \in (0,1/2)$ and $(X_1,\dots,X_n)$ represents portfolio losses. 
This is due to the strict monotonicity of $\alpha\mapsto \ex_{\alpha}(X)$ when $X$ is not a constant. 
Hence, DQ based on expectile  can be interpreted as the unique adjusted rule for the risk level of the aggregate loss, accounting for the effects of diversification. For DQs based on VaR or ES, such uniqueness result further requires some conditions on the quantile function of the aggregate risk to ensure a unique solution $c$.

\end{enumerate}




We also show that DQ based on expectiles  shares several desirable properties with DQs based on VaR and ES.
 First,   DQ based on expectiles has the same range as DQ based on ES, spanning the full range from 0 to 1, where a value of 0 indicates that the portfolio has no insolvency risk with pooled individual capital, and a value of 1 represents a portfolio relying entirely on a single asset (Proposition \ref{th:ex-01}).
Second, the value of DQ based on expectiles tends to 0 for independent portfolios as the asset number increases, aligning with intuitive diversification principles (Proposition \ref{uncorrelated}).
Finally, DQ based on expectiles has similar explicit formulas and asymptotic behavior as DQs based on VaR and ES (Section \ref{sec:5}).

The paper is organized as follows. In Section \ref{sec:2}, we recall the definition of DQ along with some properties of expectiles. Section \ref{sec:3} presents two alternative formulas and properties for DQ based on expectiles. In Section \ref{sec:4}, we focus on the portfolio selection problem using the expectile-based DQ, where we show the pseudo-convexity of the expectile-based DQ and present two efficient algorithms. In Section \ref{sec:5}, we derive the explicit formulas and the limiting behavior of DQ based on expectiles under the elliptical and multivariate regular variation (MRV) models. Several numerical illustrations and experiments for the portfolio selection problem are provided in Section \ref{sec:6}. Finally, Section \ref{sec:7} concludes the paper.



  \section{Preliminaries}\label{sec:2}

Throughout this paper, we work with a nonatomic probability space $(\Omega,\mathcal F,\p)$. All equalities and inequalities of functionals on $(\Omega,\mathcal F,\p)$ are under $\p$ almost surely ($\p$-a.s.)~sense, and $\p$-a.s.~equal random variables are treated as identical.
A risk measure  $\phi$ is a mapping  from $\X$ to $\R$, where $\X$ is  a convex cone of random variables  on  $(\Omega,\mathcal F,\p)$ representing losses faced by a financial institution or an investor, and $\X$ is assumed to include all constants (i.e., degenerate random variables). For $p\in (0,\infty)$, denote by $L^p=L^p(\Omega,\mathcal F,\p)$ the set of all random variables $X$ with $\E[|X|^p]<\infty$ where $\E$ is the expectation under $\p$.  Write $X\sim F_X$  if the random variable $X$ has the distribution function $F_X$  under $\mathbb{P}$,  and $\overline F_X=1-F_X$. Denote by 
	$X \laweq  Y$ if two random variables $X$ and $Y$ have the same distribution. 	
	We always write $\mathbf X=(X_1,\dots,X_n)$ and  $\mathbf 0 $ for the $n$-vector of zeros. 
	Further, denote by $[n]=\{1,\dots,n\}$ and $\R_+=[0,\infty)$.  Terms such as increasing or decreasing functions are in the non-strict sense. 
The diversification quotient  is defined as follows by \cite{HLW24}.
\begin{definition} Let $\rho=\left(\rho_{\alpha}\right)_{\alpha \in I}$ be a class of risk measures indexed by $\alpha \in I=$ $(0, \overline{\alpha})$ with $\overline{\alpha} \in(0, \infty]$ such that $\rho_{\alpha}$ is decreasing in $\alpha$. For $\mathbf{X} \in \mathcal{X}^{n}$, the \emph{diversification quotient} (DQ) based on the class $\rho$ at level $\alpha \in I$ is defined by $$ \mathrm{DQ}_{\alpha}^{\rho}(\mathbf{X})=\frac{\alpha^{*}}{\alpha}, \quad \text { where } \alpha^{*}=\inf \left\{\beta \in I: \rho_{\beta}\left(\sum_{i=1}^{n} X_{i}\right) \leq \sum_{i=1}^{n} \rho_{\alpha}\left(X_{i}\right)\right\} $$ with the convention $\inf (\varnothing)=\overline{\alpha}$. \label{def:1} \end{definition}
The parameterized class $\rho$ includes  commonly used risk measures such as   Value-at-Risk (VaR), the Expected Shortfall (ES), expectiles, mean-variance, and entropic risk measures. 
Among these,    VaR and ES  are the most  two popular classes of risk measures in banking and insurance practice. {We use ``small $\alpha$" convention for VaR and ES.} 
The VaR at level $\alpha \in [0,1)$ is defined as$$
	\VaR_\alpha(X)=\inf\{x\in \R: \p(X\le x) \ge 1-\alpha\},~~~X\in L^0,
	$$
	and the ES (also called CVaR, TVaR or AVaR) at level $\alpha \in (0,1)$ is defined as
	$$
	\ES_{\alpha}(X) = \frac 1 \alpha \int_{0}^\alpha \VaR_\beta(X) \d \beta,~~~X\in L^1,
	$$
	and  $\ES_0(X)=\esssup(X)=\VaR_0(X)$, which may be $\infty$.
	The probability level $\alpha$ above is typically very small, e.g., $0.01$ or $0.025$ as  in the regulatory document \cite{BASEL19}.

 The expectile at a given confident level $\alpha\in(0,1)$ for a loss random variable is the unique minimizer of an asymmetric quadratic loss. Mathematically, the expectile   of a loss random variable $X\in L^2$  at a confidence level $\alpha \in(0,1)$, denoted by $\ex_{\alpha}(X)$, is defined as the following minimizer:  \begin{equation}\label{eq:expectile}
\ex_{\alpha}(X)=\argmin_{x \in \mathbb{R}}   \left\{(1-\alpha) \mathbb{E}\left[(X-x)_{+}^{2}\right]+\alpha \mathbb{E}\left[(X-x)_{-}^{2}\right]\right\}
\end{equation}
where $(x)_+=\max\{0,x\}$ and  $(x)_+=\max\{0,-x\}$.
By the first-order condition, $\ex_{\alpha}(X)$  is equal to the unique number $t$   satisfying 
\begin{equation}\label{eq:ex} (1-\alpha) \mathbb{E}\left[(X-t)_{+}\right]=\alpha \mathbb{E}\left[(X-t)_{-}\right].
\end{equation}
The equation \eqref{eq:ex} is well-defined for each $X \in L^1$, which serves as the natural domain for expectiles. We take \eqref{eq:ex} as the definition of expectiles, and let $\X=L^1$ in the following context.

We first revisit some properties of expectiles.  From \eqref{eq:ex}, we can observe that $\alpha\mapsto \ex_\alpha(X)$ is decreasing  and it is strictly decreasing when $X$ is non-degenerate; see e.g. \citet[Theorem 1]{NP87}. Furthermore, \eqref{eq:ex}    is equivalent  to 
\begin{equation}\label{eq:2} \ex_\alpha(X)=\mathbb{E}[X]+\theta \mathbb{E}\left[(X-\ex_\alpha(X))_{+}\right] \quad \text { with } \theta=\frac{1-2 \alpha}{\alpha}, \end{equation}
due to the equality of $(x)_{+}-(x)_{-}=x$.  Since $\theta$ in \eqref{eq:2} is strictly decreasing in $\alpha \in(0,1)$ and $-1<\theta<\infty$, it follows that  $\ex_\alpha(X)=\mathbb{E}[X]$ for $\alpha=1/2$, $\ex_\alpha(X) \leq \mathbb{E}[X]$ for $\alpha>1/2$,  and $\ex_\alpha(X)  \geq \mathbb{E}[X]$ for $\alpha<1/2$, where the strict inequalities between $\ex_\alpha(X)$ and $\E[X]$ hold if $X$ is non-degenerate.
When $\alpha<1/2$, the positive $\theta$ can be interpreted as the safe loading of the upper partial moment of loss in addition to the mean of a loss. 
In the risk management of insurance and finance, a risk measure of a loss random variable $X$ is usually viewed as a premium or a regulatory capital, which is often required to be larger than the expected loss $\mathbb{E}[X]$.  Therefore, in insurance and finance, researchers are usually interested in the case with $\alpha<1/2$ or equivalently $\theta>0$.  
 We will do the same in this paper.   

 
A  risk measure that can be defined as the minimizer of a suitable expected loss function is called elicitable (\cite{G11}); see  \cite{EPRWB14} for more discussion of elicitability in risk management. 
  \cite{Z16} proved that
 $\ex_\alpha$ for  $\alpha \in(0,1/2]$ is the only class of coherent risk measures that are also elicitable, and this feature makes expectile special in the contexts of backtesting and forecast comparison, among all coherent risk measures. 
 
    The acceptance set of expectiles can be written as
\begin{equation}\label{eq:acceptance}
\mathcal{A}_{\ex_\alpha}=\left\{X ~\Bigg|~ \frac{\mathbb{E}\left[X_{-}\right]}{\mathbb{E}\left[X_{+}\right]} \geq \frac{1-\alpha}{\alpha}\right\},~~~\text{or~equivalently,}~~~\mathcal{A}_{\ex_\alpha}=\left\{X ~\Bigg|~ \frac{\mathbb{E}\left[X_{+}\right]}{\mathbb{E}\left[|X|\right]} \leq \alpha\right\};
\end{equation}
 see, e.g., Equation (4) of \cite{BD17}.   
 
 Now, we provide the main concept of our study in the definition below.
\begin{definition}\label{de:DQex}  For $\mathbf{X} \in \mathcal{X}^{n}$, DQ based on   expectiles at level $\alpha \in (0,1)$ is defined by \begin{equation}\label{eq:DQ_ex0}\mathrm{DQ}_{\alpha}^{\ex}(\mathbf{X})=\frac{\alpha^*}{\alpha},~~~~\alpha^*=\inf \left\{\beta \in (0,1): \ex_{\beta}\left(\sum_{i=1}^{n} X_{i}\right) \leq \sum_{i=1}^{n} \ex_{\alpha}\left(X_{i}\right)\right\}\end{equation}
with the convention $\inf (\varnothing)=1$.
\end{definition}

\begin{remark}
 Since $\ex_{1/2}$ is the mean and   $\ex_\beta(X)>\E[X]$ for all non-degenerate $X\in L^1$ and $\beta<1/2$, it holds that $\mathrm{DQ}^\ex_{1/2}(\mathbf X)=0$ if $\sum_{i=1}^n  X_i$ is a constant, and $\mathrm{DQ}^\ex_{1/2} (\mathbf X)=1$  otherwise.
Hence, case $\alpha=1/2$  is trivial. In the following discussion, we will exclude $\alpha=1/2$.
\end{remark}

We are mainly interested in the case $\alpha \in (0,1/2)$ throughout the paper, except for a symmetry statement in Proposition \ref{prop:2} where $\mathrm{DQ}^{\ex}_{1-\alpha}$ appears. 

As a member of DQ, $\DQex$ satisfies the following three simple properties of diversification indices.
\begin{enumerate}[(i)]
\item[{[+]}] Non-negativity:  $\DQex(\mathbf{X})\ge 0$ for all $\mathbf X\in \X^n$.
\item[{[LI]}] Location invariance: $\DQex(\mathbf{X}+\mathbf{c})=\DQex(\mathbf{X})$  for all $\mathbf{c}=(c_1,\dots,c_n) \in \R^n$ and all $ \mathbf X\in \X^n$.
\item[{[SI]}] Scale invariance: $\DQex(\lambda \mathbf{X})=\DQex(\mathbf{X})$  for all   $\lambda>0$ and all $ \mathbf X\in \X^n$.
\end{enumerate}

\section{DQ based on expectiles}\label{sec:3}
In this section, we will focus on the theoretical properties $\rm{DQ}^{\ex}_\alpha$. First of all, we present two alternative formulas of $\rm{DQ}^{\ex}_\alpha$, which provide convenience for computation and
optimization.
One of the formulas is related to  the Omega ratio introduced by \cite{KS02}, while the other shows that $\DQex$ measures the improvement of insolvency probability when risks are pooled, similar to the expression of DQs based VaR and ES in \citet[Theorem 4]{HLW24}.


\begin{theorem}\label{th:var}
 For a given $\alpha \in(0,1/2)$, ${\rm DQ}^{\ex}_\alpha $   has the alternative formula
\begin{equation}\label{eq:var-alter}
{\rm DQ}^{\ex}_\alpha (\mathbf X)  =\frac{1}{\alpha}
\frac{\mathbb{E}\left[ (S-t) _+\right]}{\mathbb{E}\left[\left| (S-t) \right|\right]},\mbox{~where }S=\sum_{i=1}^n X_i \mbox{~and~}t=\sum_{i=1}^n \ex_\alpha(X_i),~~~~\mathbf X\in \X^n.
\end{equation}
Moreover, let  $F_{S}$ be the distribution of $S$, then  ${\rm DQ}^{\ex}_\alpha (\mathbf X)$ can also been computed by
\begin{equation}\label{eq:DQ_ex}
\mathrm{DQ}^{\ex}_\alpha (\mathbf{X})=\frac{1-\widetilde{F}_{S}(t) }{\alpha},~~~\mathbf X\in \X^n,\end{equation}
where
\begin{equation}\label{t_F}
\widetilde{F}_{S}(y)=\frac{y F_{S}(y)-\int_{-\infty}^y x \mathrm{d} F_{S}(x)}{2\left(y F_{S}(y)-\int_{-\infty}^y x \mathrm{d} F_{S}(x) \right)+\E[S]-y}.
\end{equation}
\end{theorem}

The formula in \eqref{eq:var-alter} is closely related to the concept of  Omega ratio.
For a payoff $R$, the Omega ratio  is defined as \begin{equation}\label{eq:Om1}
\Omega_R(t) = \frac{\mathbb{E}[(R - t)_{+}]}{\mathbb{E}[(R - t)_{-}]},
\end{equation}
where $t\in\R$.   
This ratio compares the expected upside deviation (gains) to the expected downside deviation (losses) for a given threshold. 
Therefore, the larger Omega ratio indicates that the expected gains outweigh the expected losses, suggesting a more favorable risk-return profile.  It is clear from \eqref{eq:ex} that  
$
\Omega_X(\ex_\alpha)={\alpha}/{(1-\alpha)};
$  
that is the Omega ratio of $X$ is $ \alpha/(1-\alpha)$ when setting $t = \ex_\alpha(X)$ as the threshold. In the special case where $t = 0$, $\Omega_R(0)$ is also known as the gain-loss ratio (see \cite{BL00}). In our context, where positive 
$X$ represents risk, a lower value of 
$\Omega_X$ is preferred.
Note that the Omega ratio is defined for one random variable (either representing an individual asset or a portfolio), and it does not quantify diversification.  On the other hand, $\DQex$ takes  into account both the total loss $S$ and the marginal risks through $t=\sum_{i=1}^n \ex_\alpha(X_i)$ in \eqref{eq:var-alter}.

\begin{remark}\label{rem:omega}
Let $S_{\mathbf X}=\sum_{i=1}^n X_i$ for $\mathbf X\in \X^n$.
Note that we can rewrite $\mathrm{DQ}^{\ex}_\alpha$ with the definition of Omega ratio as
\begin{align*}
\mathrm{DQ}^{\ex}_\alpha (\mathbf{X})&=\frac{1}{\alpha(1+1/\Omega_ {S_{\mathbf X}}(\sum_{i=1}^n \ex_{\alpha}(X_i)) )}.
\end{align*} 
Thus,  for $\mathbf X, \mathbf Y\in \X^n$,  we have 
$$ \Omega_ {S_{\mathbf X}}\left(\sum_{i=1}^n \ex_{\alpha}(X_i)\right) \leq \Omega_ {S_{\mathbf Y}}\left(\sum_{i=1}^n \ex_{\alpha}(Y_i)\right) ~\Longleftrightarrow~ \mathrm{DQ}^{\ex}_\alpha (\mathbf{X}) \leq \mathrm{DQ}^{\ex}_\alpha (\mathbf{Y}).$$
  This indicates that the diversification level of a vector $\mathbf{X}$  is uniquely determined by the Omega ratio of the aggregate risk at the threshold $\sum_{i=1}^{n} \ex_{\alpha}(X_i)$.
\end{remark}

The formula in \eqref{eq:DQ_ex} can be calculated by transforming the original distribution of $S_{\mathbf X}$,  relying on the properties of expectiles introduced by  \citet[ Proposition 8.23]{MFE15}, which shows  that \begin{equation}\label{trans_p}\ex_\alpha(S_{\mathbf X}) = \widetilde{F}_{S_{\mathbf X}}^{-1}(1-\alpha),\end{equation} 
with $\widetilde{F}_{S_{\mathbf X}}$ defined in \eqref{t_F}.
It is easy to check that  $\widetilde F_{S_{\mathbf X}}$ is a continuous distribution function that strictly increases its support. Then \eqref{eq:DQ_ex} can be rewritten as 
$$
\DQex (\mathbf X) =\frac{1}{\alpha} \mathbb Q\left(S_{\mathbf X}> \sum_{i=1}^n \ex_{\alpha}(X_i)\right),~~~\mathbf X\in \X^n, 
$$
for some probability measure $\mathbb{Q}$. Similar representations of $\DQVaR$ and $\DQES$ can be found in \citet[Remark 2]{HLW23}.


The diversification ratio (DR)      based on expectiles  is defined as\footnote{If the denominator in the definition of ${\rm DR}^{\phi}(\mathbf X)$ is $0$, then we use the convention $0/0=0$ and $1/0=\infty$.}  
$$
{\rm DR}^{\ex_\alpha}(\mathbf X)= \frac{ \ex_\alpha \left(\sum_{i=1}^n X_i\right)}{ \sum_{i=1}^n \ex_\alpha(X_i)}.
$$
We can see that DQ measures the improvement of insolvency probability, and DR measures the expectile improvement. For more comparison of DQs and DRs based on VaR and ES, we refer to  Section 5 of \cite{HLW24}.

We summarize several dependence structures
that correspond to special values $0$ and  $1$ of    $\mathrm{DQ}^{\ex}_\alpha $. 

\begin{proposition}\label{th:ex-01}
 For $\alpha \in(0,1/2)$ and $n\ge 2$,   we have
 \begin{enumerate}[(i)]
  \item $\left\{\mathrm{DQ}^{\ex}_{\alpha}(\mathbf{X})\mid\mathbf{X} \in \X^n\right\}=[0,1]$.
  \item ${\rm DQ}^{\ex}_\alpha (\mathbf X)=0$ if and only if $\sum_{i=1}^ n X_i\le \sum_{i=1}^ n \ex_\alpha(X_i)$ a.s.  In case $\sum_{i=1}^ n X_i$ is a constant,    $ {\rm DQ}^{\ex}_\alpha (\mathbf X)=0$.
  \item If $\mathbf X=(\lambda_1X, \dots, \lambda_n X)$ for some $(\lambda_1, \dots, \lambda_n) \in \R^n_+\setminus \{\mathbf 0\}$ and  non-degenerate $X\in \X$, then $ {\rm DQ}^{\ex}_\alpha (\mathbf X)=1$.

\end{enumerate}
   \end{proposition}
  
Proposition \ref{th:ex-01} shows that  ${\rm DQ}^{\ex}_\alpha (\mathbf X)$
takes value on a continuous and bounded range. In contrast,  $\mathrm{DR}^{\ex_\alpha}$ falls into the range $[0,1]$ only when the expectile of the total risk is nonnegative. 
Similar to DQs based on VaR and ES, the value of $\DQex$  equals to 0 indicating that there is no insolvency risk with pooled individual capital and its
value is 1 if there is strong positive dependence such that the portfolio relies on a single asset.


In  the following, we further consider some properties of ${\rm DQ}^{\ex}_\alpha$ based on the properties of expectiles. 

\begin{proposition}\label{prop:2}
Assume $\alpha \in(0,1/2)$.
\begin{enumerate}[(i)]
\item If $\DQex(\mathbf X)>0$,  then
 $\DQex(\mathbf X)$ is the unique solution of $c$ that satisfies
\begin{align} \label{eq: exist}\ex_{c\alpha} \left(\sum_{i=1}^n X_i\right) = \sum_{i=1}^n \ex_{\alpha}(X_i).\end{align}
\item If $\mathbf X$ is non-degenerate,
\begin{equation}\label{eq:sym}
\alpha\mathrm{DQ}^{\ex}_\alpha (\mathbf{X})+ (1-\alpha)\mathrm{DQ}^{\ex}_{1-\alpha} (-\mathbf{X})=1. \end{equation}
\end{enumerate}
\end{proposition}

Note that under the mild assumption stated in Proposition \ref{prop:2} (i), we can show the existence and uniqueness of $c$ such that equation \eqref{eq: exist} holds. The conditions for the DQs based on VaR and ES require a bit further, as neither $ \ES_\alpha$ nor $ \VaR_\alpha $ are strictly decreasing in $\alpha$.    For $\mathrm{DQ}^\VaR_\alpha$, we additionally require that $\sum_{i=1}^n X_i$ has a continuous distribution. For $\mathrm{DQ}^\ES_\alpha$, it is enough that $\VaR_{\beta}(\sum_{i=1}^n X_i)$ is not constant over the interval $(0, \alpha)$.  Additionally, although we focus on  $\alpha \in (0, 1/2)$, we can easily compute the $\DQex$ for $\alpha \in (1/2, 1)$ using  \eqref{eq:sym} by symmetric property of expectiles when $\mathbf X$ is non-degenerate. If $\mathbf X$ is a constant vector, it is clear that $\DQex(\mathbf X)=0$ for all $\alpha \in (0,1)$.

An independent portfolio with more assets is widely regarded as better diversified. As $n \to \infty$, we will expect that the value of  $\DQex$ for the $n$-asset independent portfolio approaches zero.
We now examine the asymptotic behavior of DQ for such large portfolios. 
\begin{proposition}\label{uncorrelated}Let $X_1, X_2, \ldots$ be a sequence of  uncorrelated random variables in $L^2$ satisfying $\inf_{i\in \mathbb N}\{\ex_\alpha(X_i)-\E[X_i]\}>0$ and $\sup _{i \in \mathbb{N}} \mathrm{var}\left(X_i\right)<\infty$.  For $\alpha \in(0,1/2)$, we have 
$$
\lim _{n \rightarrow \infty} \mathrm{DQ}_\alpha^\ex(X_1, \ldots, X_n)=0.
$$
\end{proposition}

The special case of Proposition \ref{uncorrelated} occurs when $X_1, X_2,\ldots$ are iid. For this case, by Theorem 2.2.3 of   \cite{D19}, $X_i$ is non-degenerate  with  $\E|X_i| < \infty$ for $i\in [n]$ is sufficient for the above conclusion.

Below, we provide an example to show that  DQs based on VaR and ES can exhibit significant jumps when analyzing a portfolio that follows a Bernoulli distribution.  In contrast, $\mathrm{DQ}^{\ex}_\alpha$, incorporates information from the entire distribution, leading to more stable assessments of diversification. 
\begin{example}\label{ex:stable}
For  $i\in[n]$,  let $X_i$  be iid Bernoulli-distributed losses with parameter $ p=0.1$.  We can compute $$
\ex_\alpha(X_i) = \frac{(1-\alpha)p}{\alpha + p(1 - 2\alpha)}. 
$$
For $n=2$, we have 
$$
\ex_\beta(X_1 + X_2) =\left\{
\begin{aligned} &\frac{2p\beta + 2p^2 - 4\beta p^2}{p^2 - 2\beta p^2 + \beta},  &\beta \leq {p^2}/{(1 - 2p + 2p^2)},\\&\frac{2p - 2\beta p}{2p-p^2 - 4\beta p + 2p^2\beta+\beta},&\beta > {p^2}/{(1 - 2p + 2p^2)}.\end{aligned}\right.
$$
Then  $\alpha^*$ can be solved as
$$
\alpha^*_\ex =\left\{\begin{aligned} &\frac{p\alpha}{1-2\alpha(1-p)}, &\alpha\leq p,\\ & \frac{\alpha-p+p^2-\alpha p^2}{2p\alpha+1-3p+2p^2(1-\alpha)},& \alpha > p.\end{aligned}\right.$$  
 In contrast,  we have $\alpha^*_\VaR=0.19\id_{\{\alpha>0.1\}} $ and     $\alpha^*_\ES=\alpha/(20-100\alpha) \id_{\{0.1<\alpha\leq 18/95\}}+\alpha \id_{\{\alpha>18/95\}}.$ The results are shown in Figure \ref{fig:Bern2}. We also perform  an analysis for $
n=10$, which is presented in Figure \ref{fig:Bern10}. We can see that when $\alpha$ is less than $0.1$, both $\mathrm{DQ}_\alpha^{\mathrm{VaR}}$ and $\mathrm{DQ}_\alpha^{\mathrm{ES}}$ are zero, and as $\alpha$ increases, they provide no further useful information. However, the expectile-based DQ steadily increases with $\alpha$ for $\alpha\in(0,0.1)$.

The same issue appears in the empirical study. For example, if $\alpha =0.01$ and the dataset has less than $100$ data points, then VaR and ES are difficult to interpret, and they lead to DQ being $0$ (see Remark \ref{rem:small-sample} in Section \ref{sec:4}). In contrast, expectiles allow for the use of a broader dataset, yielding more stable results when we apply it to real data. 

\begin{figure}[t]
\centering
 \includegraphics[height=5cm]{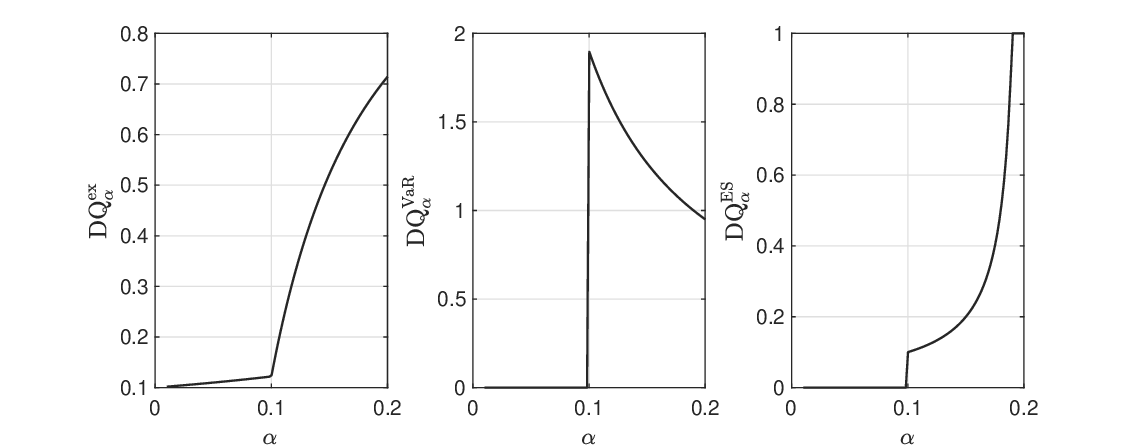}
 \captionsetup{font=small}
\caption{The values of DQ when $n=2$}\label{fig:Bern2}

\end{figure}

\begin{figure}[t]
\centering
 \includegraphics[height=5cm]{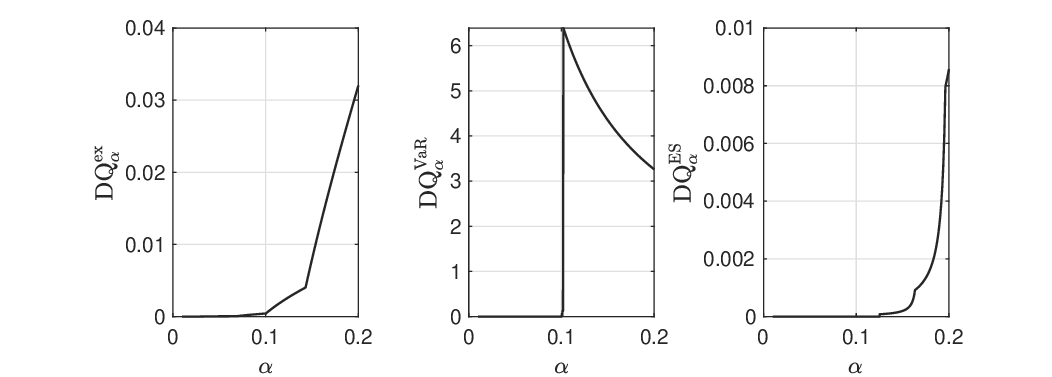}
 \captionsetup{font=small}
\caption{The values of DQ when $n=10$}\label{fig:Bern10}
\end{figure}
\end{example}

\section{Portfolio selection}\label{sec:4}

In this section, we analyze the portfolio selection problem for assets with
loss vector  $\mathbf{X}\in\X^n$ by selecting portfolio weights  $\mathbf w= (w_{1},  \dots, w_{n}) \in \Delta_{n}$  to minimize DQ based expectiles, where  $$  \Delta_{n}:=\left\{\mathbf{w} \in[0,1]^{n}: w_{1} +\dots+w_{n}=1\right\}.$$ 

We write $\mathbf w \odot \mathbf X=\left(w_1X_1,\dots,w_nX_n\right)$ as the portfolio loss vector with the weight $\mathbf w$.
For the portfolio $\mathbf w \odot \mathbf X$, the total loss is  $\mathbf{w}^\top \mathbf{X}$. For the portfolio selection problem, we now treat  
 $\mathrm{DQ}^\ex_\alpha(\mathbf w \odot \mathbf X)$ as a function of the portfolio weight $\mathbf w$. 
We focus on the following optimal diversification problem 
\begin{equation}\label{eq:optimal_DQ}
\min _{\mathbf{w} \in \Delta_n}  {\rm DQ}^{\ex}_{\alpha}(\mathbf w \odot \mathbf X).\end{equation}

\subsection{Quasi-convexity and pseudo-convexity}
 Denote by  $\mathbf{x}_{\alpha}^{\ex}=(\ex_{\alpha}(X_{1}), \ldots, \ex_{\alpha}(X_{n}))$, which   does not depend on the decision variable $\mathbf{w}.$   
Recall that 
 $$ {\rm DQ}^{\ex}_{\alpha}(\mathbf w \odot \mathbf X)= \frac{1}\alpha    \inf\left\{\beta \in (0,1) :  \ex_{\beta}\left(\sum_{i=1}^n w_iX_i\right) \le \mathbf w^\top  \mathbf{x}_{\alpha}^{\ex} \right\}.$$

 \citet[Theorem 2]{HLW24} showed that, for any decreasing family $\rho=(\rho_\alpha)_{\alpha\in I}$ of coherent risk measures, $\mathbf w \mapsto \mathrm{DQ}_\alpha^{\rho}(\mathbf w \odot \mathbf X)$  is quasi-convex, meaning that 
 $$
 \mathrm{DQ}_\alpha^{\rho}((\lambda\mathbf w+(1-\lambda)\mathbf v )\odot \mathbf X) \le \max
 \left\{\mathrm{DQ}_\alpha^{\rho}(\mathbf w \odot \mathbf X), \mathrm{DQ}_\alpha^{\ex}(\mathbf v \odot \mathbf X)\right\}
 $$
 for all $\lambda \in[0,1]$, $\mathbf w,\mathbf v\in \Delta_n$ and $\mathbf X\in \X^n$.
 Since expectiles are coherent for $\alpha \in (0,1/2)$, $\mathbf w \mapsto \mathrm{DQ}_\alpha^{\ex}(\mathbf w \odot \mathbf X)$ is quasi-convex on $\Delta_n$ for $\alpha \in(0,1/2)$.
 Note that since  $\mathrm{DQ}_\alpha^{\ex}$ is scale invariant, 
optimization of $\mathrm{DQ}_\alpha^{\ex}$ over $\Delta_n$ can be equivalently  formulated as optimization 
over $\R^n_+\setminus\{\mathbf 0\}$.
The next proposition shows that $\mathrm{DQ}_\alpha^{\ex}$ is also quasi-convexity on $\R^n_+$. 
 \begin{proposition}\label{quasi-convex}
     Let $\alpha \in (0, 1/2)$. The mapping $\mathbf w \mapsto \mathrm{DQ}_\alpha^{\ex}$ 
     is quasi-convex on $\R^n_+$.
 \end{proposition}

Our next result further shows that $\mathbf w \mapsto \mathrm{DQ}_\alpha^{\ex}(\mathbf w \odot \mathbf X)$ is pseudo-convex in $\mathbf{w}$.
Pseudo-convexity is stronger than quasi-convexity, and it is a useful property in optimization, allowing for gradient descent algorithms. 
A differentiable function $f: \mathbb{T} \rightarrow \mathbb{R}$, where $\mathbb{T} \subseteq \mathbb{R}^n$ is an open set, is said to be \emph{pseudo-convex} if 
\begin{equation}
\label{eq:pseudo}
\text{for all } \mathbf{x}, \mathbf{y} \in \mathbb{T}: \quad f(\mathbf{y}) < f(\mathbf{x}) \Longrightarrow \nabla f(\mathbf{x}) \cdot (\mathbf{y} - \mathbf{x}) < 0.
\end{equation}
Here, $\nabla f$ is the gradient of $f$, that is, $$\nabla f=\left(\frac{\partial f}{\partial x_1}, \ldots, \frac{\partial f}{\partial x_n}\right).$$
We can also speak of pseudo-convexity of $f$ on a  (not necessarily open) subset $\mathbb T'$ of $\mathbb T$, which only requires
\eqref{eq:pseudo} to hold on $\mathbb T'$. 
For instance, if $f$ is pseudo-convex on $(0,\infty)^n$, then it is also pseudo-convex on the interior of $\Delta_n$.

For pseudo-convex functions,    any local minimum is also a global minimum (\cite{M75}). 

\begin{theorem} \label{prop:pseudo}Let $\alpha \in (0, 1/2)$. The mapping $f:\mathbf{w} \mapsto {\rm DQ}^{\ex}_{\alpha}(\mathbf{w} \odot \mathbf{X})$ on $(0,\infty)^n$ is  pseudo-convex. 
\end{theorem}

 Below we show that the optimization problem is similar to the mean-variance problem.
\begin{proposition}\label{thm:2}Fix $\alpha\in(0,1/2)$ and $\mathbf X \in \X^n$.
The optimization of $\mathrm{DQ}^\ex_\alpha(\mathbf w \odot \mathbf X)$ in \eqref{eq:optimal_DQ}  is equivalent to the following two optimization problems:
\begin{equation}\label{optimal}
 \mathrm{minimize}~ \frac{\mathbb{E}[(\mathbf  w^\top(\mathbf X- \mathbf{x}_{\alpha}^{\ex} ))_+]}{\mathbb{E}\left[ \mathbf  w^\top (\mathbf{x}_{\alpha}^{\ex} - \mathbf X)\right]} ~~~\mathrm{over~}\mathbf w \in \Delta_n\end{equation}
 and
 \begin{align}\label{eq:opt3} 
 \begin{aligned}
 &  \mathrm{minimize}~  \frac{\mathbb{E}\left[\left(\mathbf  w^\top(\mathbf X- \mathbf{x}_{\alpha}^{\ex} )\right)_+\right]}{m}~~~\mathrm{over~}\mathbf{w} \in \Delta_n
  \\&\mathrm{subject ~to}~\mathbb{E}\left[\mathbf  w^\top(\mathbf{x}_{\alpha}^{\ex}-\mathbf X )\right]=m; ~m> 0.     
 \end{aligned}\end{align} 
\end{proposition}

\begin{remark}
We note that neither $\mathbf X \mapsto \DQex(\mathbf X)$ nor $\mathbf w \mapsto \DQex(\mathbf w \odot \mathbf X)$ is convex. As discussed in \citet[Appendix C.2]{HLW24}, convexity is not a desirable property for a diversification index. 
Quasi-convexity and pseudo-convexity of  $\mathbf w \mapsto \DQex(\mathbf w \odot \mathbf X)$ are desirable properties both for the interpretation of diversification and for feasibility in optimization.
\end{remark}

\subsection{Optimization from empirical sample}

Next, we propose efficient algorithms for optimizing $ \mathrm{DQ}^\ex_\alpha$ from empirical sample.  Assume we have data $\mathbf X^{(1)},\dots,\mathbf X^{(N)}$ sampled from $\mathbf X$ satisfying some ergodicity condition (being iid would be sufficient).
 Denote by $\widehat{\mathbf x}^{\ex}_\alpha=(\widehat \ex_{\alpha,1}(N),\dots,\widehat \ex_{\alpha, n}(N))$  the empirical estimator of $\mathbf x^{\ex}_\alpha$.  We apply the empirical estimator from   \cite{KZ17} and \cite{DGS19} to estimate $\mathbf x^{\ex}_\alpha$; that is
 $$\widehat {\ex}_{\alpha,i}(N)=\argmin_{x \in \R} \frac{1}{N}\sum_{j=1}^N\left((1-\alpha)(X^{(j)}_i-x)_+^2+\alpha(X^{(j)}_i-x)_-^2\right), ~~ i\in [n].$$

The optimization problem \eqref{optimal} can be solved by linear programming (LP), similar to the approach used in optimizing the Omega ratio problem; see e.g., \cite{KZCR14a} and \cite{GMOS16}. More specifically,  the empirical version of the problem \eqref{optimal}   can be  written as 
\begin{align*}
     \mbox{minimize ~~~}  & \frac{ \sum_{j=1}^N  (\mathbf w ^\top  (\mathbf X^{(j)}- \widehat{\mathbf x}^{\ex}_\alpha))_+}{    \sum_{j=1}^N \mathbf  w^\top  ( \widehat{\mathbf x}^{\ex}_\alpha-\mathbf X^{(j)}  )}~~~
     \\
~\text { subject to}~~~&  \sum_{i=1}^n w_i=1, 
~~w_i\geq0 \mbox{~for all $i\in [n]$}.
\end{align*}

Let  $\boldsymbol\mu=  \sum_{j=1}^N ( \widehat{\mathbf x}^{\ex}_\alpha-\mathbf X^{(j)})\in\R^n$. The problem can be written as  \begin{align*}
\mbox{~~~minimize~}  &~~ \frac{ z_1}{z}~~~    \\ \text {subject to}& ~\sum_{i=1}^n w_i=1,  ~~w_i\geq0  \mbox{~for all $i\in [n]$,}\\~~~~&\sum_{i=1}^n  w_i\mu_i =z,    \\&   \sum_{i=1}^n     w_i  ( X^{(j)}_i- \widehat{\ex}_{\alpha,i}(N)) = y_j  ,  ~~   \mbox{~for all $j\in [N]$,}  \\& \sum_{j=1}^N d_j =z_1,~  d_j \geq y_j, ~~d_j\geq0,~~  \mbox{~for all $j\in [N]$} . 
\end{align*}
 Since $\alpha <1/2$,  we have $\widehat {\ex}_{\alpha,i}> \sum_{j=1}^N X^{(j)}_i/N$; that is, $z>0$. In order to get a linear optimization problem, we further introduce variables $v=z_1/z$ and $v_0=1/z$.  Additionally, we divide all the constraints by $z$ and make the substitutions $\widetilde{d}_j=d_j / z$ and $\widetilde{y}_j=y_j / z$ for $j=1, \ldots, N$, as well as $\widetilde w_i=w_i / z$, for $i=1, \ldots, n$. Eventually, we get the following LP formulation to minimize $v$:
 \begin{equation}\label{LP}
\begin{aligned}
 \text {minimize}~~&~ v  \\ \text { subject to } &\sum_{i=1}^n \widetilde w_i=v_0, \quad v_0 \leq M, \quad \widetilde w_i \geq 0 \quad \text { for } i\in[n], \\
& \sum_{i=1}^n  \widetilde{w}_i\mu_i=1, \\
& \sum_{i=1}^n \widetilde w_i (X_i^{(j)}-\widehat \ex_{\alpha,i}(N))=\widetilde{y}_j,  ~~~~\text { for } j\in[N],\\
& \sum_{j=1}^N \widetilde{d}_j =v,  \widetilde{d}_j \geq  \widetilde{y}_j, ~\widetilde{d}_j \geq 0, ~~\text { for } j\in[N].
\end{aligned}
\end{equation}

Note that the optimization problem \eqref{eq:opt3} exhibits parallels to the mean-variance framework, but is not a convex problem.
The global optimality can be obtained by solving a series of convex problems under different $m>0$, which yields a Pareto frontier for $m$ and  $\mathbb{E}[(\mathbf  w^\top(\mathbf X- \mathbf{x}_{\alpha}^{\ex} ))_+]$. Thus, minimizing $\mathrm{DQ}^{\ex}_\alpha (\mathbf w \odot \mathbf{X}) $, subject to the constraint $\mathbf{w} \in \Delta_n$, corresponds to finding the line with the maximum slope that passes through the origin and a point on the frontier; see Figure \ref{fig:frontier}.
\begin{figure}[htb!]
\caption{The  frontier from the  minimization of  \eqref{eq:opt3}.}\label{fig:frontier}
\centering
 \includegraphics[height=5cm]{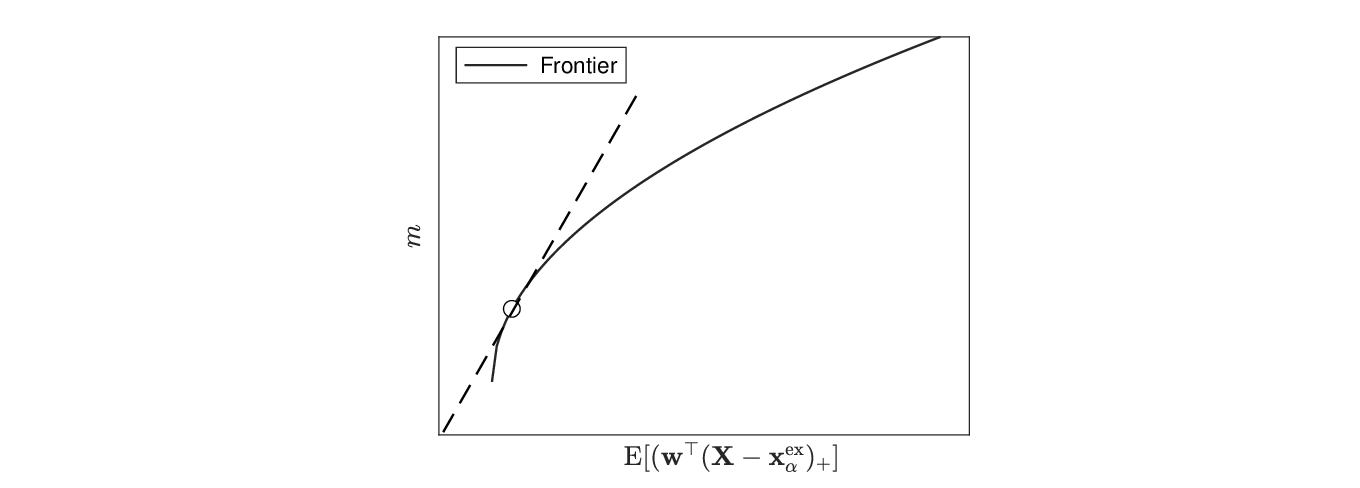}
 \captionsetup{font=small}
\end{figure}
The empirical version of  optimization problem \eqref{eq:opt3} is
\begin{equation} \label{eq:opt3_em}
\begin{aligned}
  &\mathrm{minimize}~  {\frac{1}{N m}\sum_{j=1}^N \left[\left(\mathbf  w^\top(\mathbf X^{(j)}- \widehat {\mathbf{x}}_{\alpha}^{\ex})\right)_+\right]},\\
  &~~\mathrm{subject ~to}~\frac{1}{N}\sum_{j=1}^N\left[\mathbf  w^\top(\widehat{\mathbf{x}}_{\alpha}^{\ex}-\mathbf X^{(j)} )\right]=m, ~m> 0,~\mathbf{w} \in \Delta_n,
   \end{aligned}
\end{equation} 
We can solve the empirical optimization problem with modern optimization programs such as the \texttt{fmincon} function in \texttt{Matlab}.
In Section \ref{sec:6}, we will apply \eqref{LP} to real data to explore the optimal investment strategies for stocks. It can be easily shown that the optimization results from \eqref{LP} and \eqref{eq:opt3_em} are equivalent.

\begin{remark}    \label{rem:small-sample}
For a given set of data $\mathbf X^{(1)}, \dots, \mathbf X^{(N)}$, the empirical VaR is defined as the sample quantile for $i\in [n]$, which is given by
$$\widehat{\VaR}_\alpha^i=X_i^{[k]}, \alpha \in \left(\frac{N-k}{N},\frac{N-k+1}{N}\right]$$
where $X_i^{[1]} \le \dots \le X_i^{[N]} $ are the order statistics of the data $X_i^{(1)}, \dots, X_i^{(N)}$.
The empirical ES is given by 
$$\widehat{\ES}_\alpha^i=\frac{1}{\alpha}\int_0^\alpha \widehat{\VaR}_p \d p.$$
It is clear that $\widehat{\VaR}_\alpha^i=\widehat{\ES}_\alpha^i=X_i^{[N]}$ for $\alpha \in (0,1/N)$ and $i\in [n]$. 
Therefore, $\sum_{i=1}^n X^{(j)}_i\le \sum_{i=1}^n \widehat{\VaR}_\alpha^i=\sum_{i=1}^n\widehat{\ES}_\alpha^i$ for all $j \in [N]$. Following Theorem 4 in \cite{HLW24}, the empirical $\DQVaR$ and $\DQES$ are 0.
    The same does not apply to $\DQex$ because $\widehat {\ex}_{\alpha,i}$ is strictly monotone in $\alpha$ on $(0,1/2)$ even if the data only have few points.
\end{remark}

\section{Elliptical and MRV models}\label{sec:5}

In this section, we focus on two widely-used models in finance and insurance, namely, 
  elliptical and  multivariate regular variation (MRV)  distributions.   Elliptical distributions, which encompass multivariate normal and t-distributions as specific instances, serve as fundamental tools in quantitative risk management (\cite{MFE15}). Meanwhile, the MRV model plays a significant role in Extreme Value Theory, particularly for analyzing portfolio diversification, as illustrated in works such as \cite{MR10}, \cite{ME13}, and \cite{BMWW16}. 
  For DQ based on ES and VaR, these two classes of distributions are studied by \cite{HLW23}.

\subsection{Elliptical models}
  We first define elliptical distributions. 
A random vector $\mathbf{X}$ is \emph{elliptically distributed} if
	its characteristic function can be written as
	$$
	\begin{aligned}
		\psi(\mathbf{t}) =\mathbb{E}\left[\exp \left(\texttt{i} \mathbf{t}^\top \mathbf{X}\right)\right] & 
		=\exp \left(\texttt{i} \mathbf{t}^\top \boldsymbol{\mu}\right) \tau\left(\mathbf{t}^\top \Sigma \mathbf{t}\right),
	\end{aligned}
	$$
	for some  $\boldsymbol{\mu}\in \mathbb{R}^{n}$, positive semi-definite matrix $ \Sigma\in \R^{n\times n}$,
	and $\tau: \mathbb{R}_{+} \rightarrow \mathbb{R}$ called the characteristic generator.
	We denote this distribution
	by $ \mathrm{E}_{n}(\boldsymbol{\mu}, \Sigma, \tau).
	$ We will assume that $\Sigma$ is not a  matrix of zeros. 	 Each marginal distribution of an elliptical distribution is a one-dimensional elliptical distribution with the same characteristic generator. For  a positive semi-definite matrix $\Sigma$,
	we write $\Sigma=(\sigma_{ij})_{n\times n}$, $\sigma_i^2=\sigma_{ii}$, and $\boldsymbol \sigma=(\sigma_1,\dots,\sigma_n)$.  
We define	
	\begin{equation}\label{eq:k}k_\Sigma= \frac {\sum_{i=1}^n 
			\left(\mathbf{e}^\top_i \Sigma \mathbf{e}_i \right)^{1/2}} {\left( \mathbf{1}^\top \Sigma \mathbf{1}\right)^{1/2}  } 
	=\frac{\sum_{i=1}^n\sigma_{i} }{ \left(\sum_{i, j}^n \sigma_{ij}\right)^{1/2} }
	\in [1,\infty),\end{equation}
where   $\mathbf{1}=(1,\dots,1)\in\R^n$ and  $ \mathbf e_{1},\dots, \mathbf e_{n}$ are the column vectors of the $n\times n$ identity matrix $I_n$. Moreover, $k_\Sigma = 1$ if and only if $\Sigma =\boldsymbol \sigma \boldsymbol \sigma^\top  $, which means that $\mathbf X\sim \mathrm{E}_n( \boldsymbol{\mu}, \Sigma,\tau)$ is comonotonic.

In the next proposition, we give  the explicit formula for   DQ based on expectiles.  
\begin{proposition}
	\label{prop:ellip}
	Suppose that $\mathbf X \sim  \mathrm{E}_{n}(\boldsymbol{\mu}, \Sigma, \tau)$.  We have, for $\alpha \in (0,1/2)$,
	\begin{equation} \label{DQ_elli}{\rm DQ}_\alpha^{\ex}(\mathbf X)
	=\frac{\E[(Y-k_\Sigma \ex_\alpha(Y))_+]}{\alpha\E[\vert Y-k_\Sigma \ex_\alpha(Y)\vert ]}=
		\frac{1-\widetilde F_Y  (k_\Sigma \ex_{\alpha}(Y)  )}{\alpha},\end{equation}
		where $ Y \sim \mathrm{E}_1(0,1,\tau)$ with  distribution function  $F_Y$, and  $ \widetilde F_Y$ is  defined by \eqref{t_F} by replacing $S_n$ with $Y$. 
		Moreover,
			 $k_\Sigma \mapsto {\rm DQ}_\alpha^{\ex}(\mathbf X)$ is decreasing for $\alpha \in (0,1/2)$, and  ${\rm DQ}_\alpha^{\ex}(\mathbf X)=1$  as $k_\Sigma=1$. \end{proposition}	

Recall from Proposition 2 of \cite{HLW23} that  for  $\mathbf X \sim  \mathrm{E}_{n}(\boldsymbol{\mu}, \Sigma, \tau)$,  we have, for $\alpha \in (0,1)$,
	$${\rm DQ}_\alpha^{\VaR}(\mathbf X)
	=
		\frac{1-F_Y  (k_\Sigma \VaR_{\alpha}(Y)  )}{\alpha}
		~~\text{and}~~
		{\rm DQ}_\alpha^{\ES}(\mathbf X)=
		\frac{1-    \widehat F_Y (k_\Sigma \ES_{\alpha}(Y)  )}{\alpha},		
		$$
		where    $ \widehat F_Y$ is the superquantile transform of  $F_Y$.\footnote{The \emph{superquantile transform}   of a distribution $F$ with finite mean
	is  a distribution $\widehat F_Y$ with quantile function $p\mapsto \ES_{1-p}(Y)$ for $p\in(0,1)$, where $Y\sim F_Y$.}
Contrasting DQ with DR, for $\mathbf X \sim  \mathrm{E}_{n}(\mathbf{0}, \Sigma, \tau)$ and $\alpha\in(0,1/2)$, we have
	\begin{equation}
		\label{eq:DRellip}
		{\rm DR}^{\ex_\alpha}(\mathbf X)=\frac{\ex_{\alpha}(\sum_{i=1}^n X_i)}{\sum_{i=1}^n \ex_{\alpha}(X_i)}=\frac{ \left(\sum_{i, j}^n \sigma_{ij}\right)^{1/2}   \ex_\alpha (Y) 
		}{\sum_{i=1}^n \sigma_i \ex_\alpha (Y)} = \frac{1}{k_\Sigma}={\rm DR}^{\VaR_\alpha}(\mathbf X)={\rm DR}^{\ES_\alpha}(\mathbf X).
	\end{equation}  
Thus, the three DR indices have the same value. 
For the elliptical models,   DQ depends on the selection of risk measures, effectively capturing tail risk across different distributions along with the selected $\alpha$, while  DR does not depend on $\alpha$, the distribution of $\mathbf X$, or the specific risk measure employed; it only depends on $\Sigma$. Consequently,   DQ offers a more nuanced understanding of how diversification impacts different segments of a distribution, extending beyond the conventional mean-variance framework.

	\begin{example}

We assume that  $\mathbf{X}$ follows a multivariate normal distribution $\mathrm{N}(\mathbf{0}, I_n)$ and $\mathbf{Y}$ follows a multivariate $t$-distribution $\mathrm{t}(\nu, \mathbf{0}, I_n)$, where $\mathbf{0}$ is the mean vector and $I_n$ is the $n \times n$ identity matrix. Additionally, we  assume that $\mathbf Z\sim \mathrm{it}_n(\mu)$, representing  a  joint distribution with independent t-marginals $\mathrm t(\mu,0,1)$.   We set $n=10$ and the degrees of freedom for the $t$-distribution is $\nu = 3$. The corresponding curves of DQ  based on  expectiles are shown in Figure \ref{fig3}. 
We note that   $t$-model is always less diversified compared to the normal model due to the heavier tail and common shock of $t$-distribution.  Further, since $ \mathbf{Y} $ exhibits tail dependence (as noted in Example 7.39 of \cite{MFE15}) while $ \mathbf{Z} $ does not,   large losses across the components of $ \mathbf{Y} $ are more likely to occur simultaneously compared to $ \mathbf{Z} $. We can see that for $ \alpha < 0.2 $, $\DQex$ indicates that the portfolio $ \mathbf{Z} $ is less diversified than for $ \mathbf{X}$, highlighting a nice feature of  $\mathrm{DQ}^\mathrm{ex}_\alpha$  in capturing  heavy tails. In contrast, ${\rm DR}^{\ex_\alpha}(\mathbf{X}) ={\rm DR}^{\ex_\alpha}(\mathbf{Y})=1/ \sqrt{10}$ for any $\alpha\in(0,1)$, and we can compute  ${\rm DR}^{\ex}_{0.05}(\mathbf{Z})= 0.3244>1/ \sqrt{10}$, suggesting  DR  reports a smaller diversification in iid t-model than in the common shock t-model even when $\alpha$ is small. This observation is counter-intuitive.   

Moreover, as $\alpha$ approaches 0.5, $\mathrm{DQ}_\alpha^{\mathrm{ex}}$ converges to 1 due to the expectiles aligning with the expectations.


\begin{figure}[htb!]
\caption{$\mathrm{DQ}^{\ex}_\alpha (\mathbf{X})$ for $\alpha\in(0,0.5)$.}\label{fig3}
\centering
 \includegraphics[height=5cm]{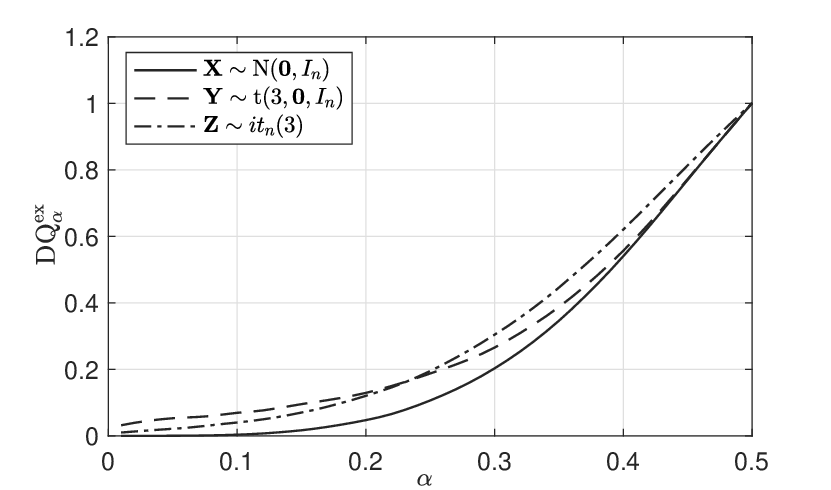}
 \captionsetup{font=small}
\end{figure}
\end{example}
In  the next example,  we show that the performance of DQs is consistent across different classes of risk measures, including VaR, ES, and expectiles.  	

\begin{example}
Let  $\mathbf{X}\sim \mathrm  t(\nu,\boldsymbol{\mu},\Sigma)$, where $
	\Sigma=(\sigma_{ij})_{4\times 4}$  with  $\sigma_{ii}=1$ and $\sigma_{ij}=0.3$ for $i\ne j$. For comparability, we use $\alpha=0.05$ for VaR, $\alpha=0.159$ for ES, and $\alpha=0.0311$ for expectile,\footnote{The parameters are set by $\ex_{0.031}(X)=\ES_ {0.159}(X)=\VaR_{0.05}(X)$ for  $X\sim \mathrm t(3)$.} and vary the value of $\nu$.
	  Note that the location $\boldsymbol{\mu}$ does not matter in computing DQ, and we can simply take $\boldsymbol \mu=\mathbf 0$. The curves for DQs based on VaR, ES, and expectiles exhibit similar behavior across the chosen probability levels, which suggests that the diversification level increases when the tail becomes thinner (larger value of $\nu$).
   \begin{figure}[htb!]
   \caption{DQs based on expectiles, ES and VaR for $\nu\in(1,10)$.}\label{fig4}
\centering
 \includegraphics[height=5cm]{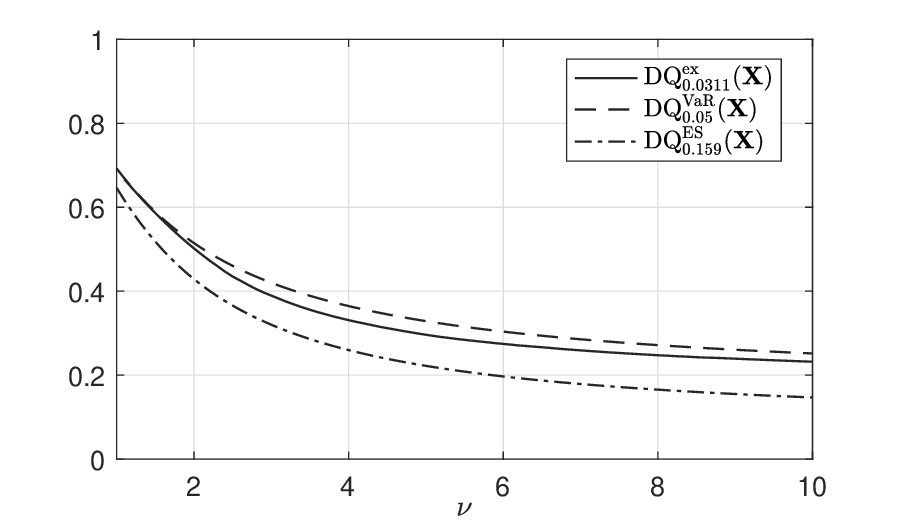}
 \captionsetup{font=small}
\end{figure}

\end{example}
In next proposition, we  show that  the portfolio optimization problem \eqref{eq:optimal_DQ} for  elliptical distributions is equal to solving the following problem 
$$\max_{\mathbf w\in \Delta_n}  k_{\mathbf w \Sigma \mathbf w^\top}:=\frac{\mathbf w^\top \boldsymbol{\sigma}}{ \sqrt{\mathbf w^\top\Sigma \mathbf w}}.$$

\begin{proposition}\label{prop:7} Suppose that $\mathbf X \sim  \mathrm{E}_{n}(\boldsymbol{\mu}, \Sigma, \tau)$.  We have, for $\alpha \in(0,1/2)$, 	\begin{align*}
\mathrm{DQ}^{\ex}_\alpha (\mathbf w \odot \mathbf X) =\frac{\mathbb{E}\left[\left(Y-k_{\mathbf w \Sigma \mathbf w^\top}\ex_\alpha(Y)\right)_+\right]}{\alpha \mathbb{E}\left[\left|Y-k_{\mathbf w \Sigma \mathbf w^\top}\ex_\alpha(Y)\right|\right]},\end{align*}	
 where   $ Y \sim \mathrm{E}_1(0,1,\tau)$; and   $$\min_{\mathbf w \in \Delta_n} {\rm DQ}^{\ex}_{\alpha}(\mathbf w \odot \mathbf X)={\rm DQ}^{\ex}_{\alpha}(\mathbf w^* \odot \mathbf X),$$  where 
 \begin{equation}\label{eq:opt-w}
		\mathbf w^*=\argmax_{\mathbf w\in \Delta_n}k_{\mathbf w \Sigma \mathbf w^\top} =\argmax_{\mathbf w\in \Delta_n} \frac{\mathbf w^\top \boldsymbol{\sigma}}{ \sqrt{\mathbf w^\top\Sigma \mathbf w}}.
	\end{equation}
\end{proposition}

\subsection{MRV models}

 A random vector $\mathbf X\in\X^n$ has an  MRV model  with some $\gamma >0$ if there exists a Borel probability measure $\Psi$ on the unit sphere $\mathbb {S}^{n-1}:=\left\{\mathbf{s} \in \mathbb{R}^{n}:\|\mathbf{s}\|=1\right\}$
		such that for any $t>0$ and any Borel set $S \subseteq \mathbb{S}^{n}$ with $\Psi(\partial S)=0$,
		$$
		\lim _{x \rightarrow \infty} \frac{\p (\|\mathbf{X}\|>t x, ~\mathbf{X}/\|\mathbf{X}\| \in S)}{\p (\|\mathbf{X}\|>x)}=t^{-\gamma} \Psi(  S), $$	where $\|\cdot\|$  is the $L_1$-norm (one could use any other norm equivalent to the $L_1$-norm). We call  $\gamma$ the tail index of  $\mathbf X$ and $\Psi$  the spectral measure of $\mathbf X$.  This is written as $\mathbf X \in \mathrm{MRV}_{\gamma}(\Psi)$.

    The univariate regular variation   with   tail index $\gamma$ is defined as 
 	$$
 	\mbox{for all~} t>0, ~\lim _{x \rightarrow \infty} \frac{1-F_{X}(t x)}{1-F_{X}(x)}=t^{-\gamma},
 	$$
 where $F_X$ is the distribution function of $X$.
We write $X\in {\rm R  V}_{\gamma}$ for this property. 	

Since $\ex_\alpha(X)/\VaR_\alpha (X)\to (\gamma-1)^{-1/\gamma}$ as $\alpha \downarrow 0$ for $X\in \mathrm{RV}_\gamma$ with $\gamma>1$ (Proposition 2.3 of \cite{BD17}), the limiting behavior of DQ based on expectiles can be derived  by using the results of Theorem 3  in \cite{HLW23}.  For the sake of completeness, we provide a separate proof below. 
\begin{proposition}\label{prop:lim}
		Suppose that  $\mathbf{X} \in \mathrm{MRV}_{\gamma}(\Psi)$  with $\gamma>1$. Then we have
		$$
		\lim _{\alpha \downarrow 0} {\rm DQ}^{\ex}_{\alpha}(\mathbf X)=\eta_{\mathbf{1}}\left(\sum_{i=1}^{n}  \eta_{\mathbf{e}_{i}}^{1 / \gamma}\right)^{-\gamma}
 ,$$ where   $\eta_{\mathbf{x}}=\int_{\mathbb{S}^{n}}\left(\mathbf{x}^\top \mathbf{s}\right)_+^{\gamma} \Psi(\d \mathbf{s})$ for $\mathbf x\in \R^n$.   Moreover, if $X_1, \dots, X_n$ are iid  with $X_1\in\mathrm {RV}_\gamma$, then $\mathrm{DQ}_\alpha^{\ex}(\mathbf X) \to n^{1-\gamma}$ as $\alpha \downarrow 0$.  	\end{proposition}

 It is well known that (see, e.g., \cite{MY15})$$
\lim _{\alpha \downarrow 0} \mathrm{DR}^{\ex_\alpha}(\mathbf X)=\lim _{\alpha \downarrow 0} \mathrm{DR}^{\VaR_\alpha}(\mathbf X)=\lim _{\alpha \downarrow 0} \mathrm{DR}^{\ES_\alpha}(\mathbf X)=n^{1 / \gamma-1}, $$
provided that $X_1, \ldots, X_n$ are iid random variables and have regularly varying tails with regularity index $\gamma>0$.  Under the same conditions, together with Theorem 3 of \cite{HLW23},  we have shown that 
$$
\lim _{\alpha \downarrow 0} \mathrm{DQ}^{\ex}_\alpha(\mathbf X)=\lim _{\alpha \downarrow 0} \mathrm{DQ}^{\VaR}_\alpha(\mathbf X)=\lim _{\alpha \downarrow 0} \mathrm{DQ}^{\ES}_\alpha(\mathbf X)=n^{1-\gamma}. $$

\section{Numerical illustrations}\label{sec:6} 	

 In this section, we apply real data to illustrate the performance of different diversification indices. In Section \ref{sec:61}, we compare the performance of DQs based on VaR, ES, expectiles, and loss-gain ratio constructed from Omega ratio for the portfolio with 5 stocks from 5 different sectors in S\&P 500. In Section \ref{sec:62}, we discuss the portfolio selection problem with 40 stocks from  10 different sectors in S\&P 500. We compare the performance of  optimally diversified portfolios obtained by minimizing DQs based on expectiles, VaR and ES, as well as maximizing Omega ratio.

 \subsection{Comparing diversification indices }\label{sec:61}
  We first identify the largest stock in each of the 10 sectors in  S\&P 500  ranked by market cap in 2014. Among these stocks, we select the 5 largest stocks  to build our portfolio\footnote{XOM from ENR, AAPL from IT, BRK/B from FINL,  WMT from  CONS, and GE from INDU.}. We compute $\mathrm{DQ}_\alpha^{\mathrm{VaR}}$, $\mathrm{DQ}_\alpha^{\mathrm{ES}}$,  and $\mathrm{DQ}_\alpha^{\mathrm{ex}}$  on each day using the empirical distribution in a rolling window of 500 days with $\alpha=0.01$ or $0.05$. The values of $\alpha$ are chosen consistent with common levels used for VaR and ES, although for expectiles these levels are not special.
  
  The historical   asset prices are from Yahoo Finance and we  use the period from January 2, 2014,  to December 29, 2023, with a total of 2516 observations of daily losses  and   500 trading days for the initial training.  
  As we discuss in Remark \ref{rem:omega}, $\DQex$ is uniquely determined by the Omega ratio. Hence, we also present the gain-loss ratio for the portfolio; that is, $\Omega=\E[S_+]/\E[S_-]$ with $S$ as the pooling loss of the portfolio. Since we use the convention of small value to represent a better diversification level, we plot the loss-gain ratio $1/\Omega$ here for comparability.

   \begin{figure}[htb!]
   \caption{DQs  based on expectiles, ES and VaR  with $\alpha = 0.01$ (left panel) and $\alpha=0.05$ (right panel), and the  loss-gain ratio.}\label{fig5}
\centering
 \includegraphics[height=5cm]{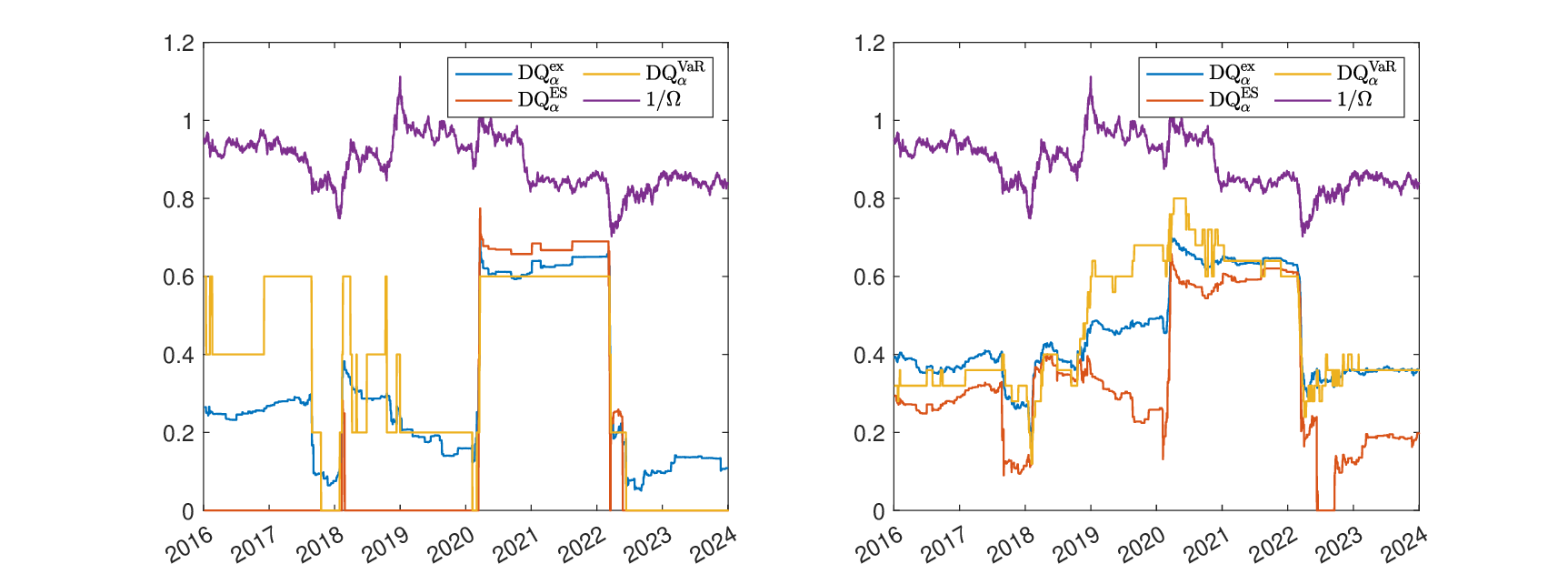}
 \captionsetup{font=small}
\end{figure}

We examine two scenarios with $\alpha = 0.01$ and $\alpha = 0.05$. As shown in Figure \ref{fig5}, the values of the DQ range from 0 to 1, with a significant increase corresponding to the onset of COVID-19. For $\alpha = 0.01$, DQs based on VaR and ES exhibit instability due to the data scarcity issue, resulting in sharp fluctuations at some time points. In contrast, DQ based on expectiles appears to be smoother, indicating a more robust performance in this context. When $\alpha = 0.05$, DQ maintains a relatively consistent temporal pattern across the three risk measures since we have a relatively larger size of tail data. Although the loss-gain ratio $1/\Omega$ does not suffer from the the data scarcity issue, it does not reflect diversification as the other three indices. For instance, in early 2020, due to COVID outbreak, the empirical correlation between assets increased substantially, indicating a weaker diversification. This is captured by all three DQs, but not the Omega ratio, which by definition only evaluates the total risk.

\subsection{Optimal diversified portfolios}\label{sec:62}
We fix $\alpha = 0.05$ and select the best diversified portfolios by minimizing $\mathrm{DQ}_\alpha^{\mathrm{ex}}$, $\mathrm{DQ}_\alpha^{\mathrm{ES}}$, and $\mathrm{DQ}_\alpha^{\mathrm{VaR}}$, as well as by maximizing the Omega ratio $\Omega_R(t_0)$  defined in \eqref{eq:Om1} with the random return $R=- \mathbf{w}^\top \mathbf{X}$ and $t_0=-\E[\sum_{i=1}^nX_i/n]$.\footnote{The threshold $t_0$ is the expected return rate of the equal weight portfolio with $\mathbf w=(1/n, \dots, 1/n)$.}
We apply the algorithm in \eqref{LP} to optimize $\mathrm{DQ}_\alpha^{\mathrm{ex}}$,  the algorithms from Proposition 6 in \cite{HLW24} to optimize $\mathrm{DQ}_\alpha^{\mathrm{VaR}}$ and $\mathrm{DQ}_\alpha^{\mathrm{ES}}$, and the algorithm from \cite{KZCR14a} to optimize the Omega ratio. These algorithms are either linear or convex optimization problems, ensuring computational efficiency.

 We  choose the 4 largest stocks from each of the 10 different sectors of S\&P 500   ranked by market cap in 2014 as the portfolio compositions (40 stocks in total). We use data from January 2, 2014,  to December 28, 2023 to build up the portfolio,    with the risk-free rate $r = 2.13\%$, corresponding to the 10-year yield of the U.S. Treasury Bill in January 2016. The initial wealth is set to 1. Our portfolio strategies are rebalanced at the start of each month, assuming no transaction costs. The optimal portfolio weights are calculated with a rolling window of 500 days at each period.
 
The portfolio performance is reported in Figure \ref{sec_40}.  Summary statistics,  including the annualized return (AR), the annualized volatility (AV),  the Sharpe ratio (SR) are reported in Table \ref{tab_40}. 
\begin{figure}[!htbp]
\caption{Wealth processes for  portfolios, 40 stocks, Jan 2016 - Dec 2023 with $\alpha=0.05$.}\label{sec_40} 
\centering
           \includegraphics[height=5cm]{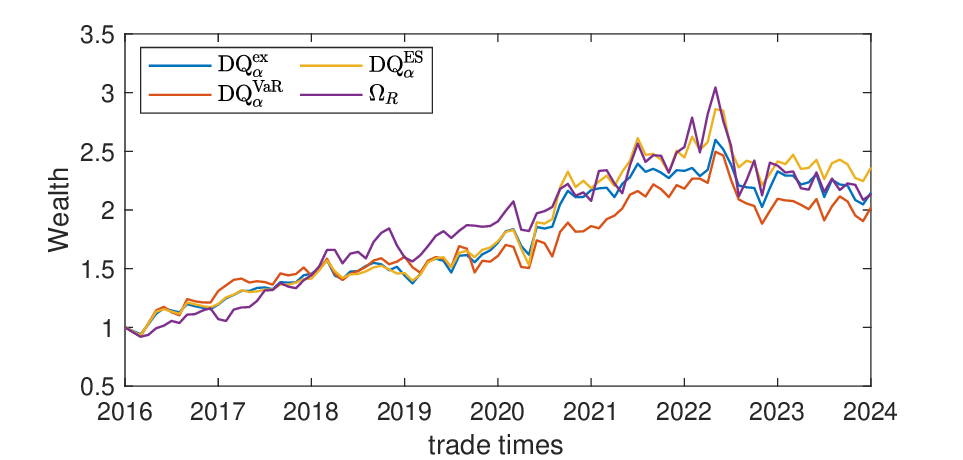}
\end{figure} 
  \begin{table}[!htbp]  
\def\arraystretch{1}
  \begin{center} 
    \caption{Summary statistics  for   different portfolio strategies from Jan 2016 to Dec 2023 with $\alpha=0.05$.} \label{tab_40} 
  \begin{tabular}{c|cccc} 
    $\%$ &${\rm DQ}^{\ex}_\alpha$ &${\rm DQ}^{\VaR}_\alpha$ &${\rm DQ}^{\ES}_\alpha$ & $\Omega_R(t_0)$  \\    \hline
AR & 10.19   &   9.35  &\textbf{11.24} & 10.07   \\
AV & \textbf{15.64}  & 18.33 & 17.88 & 20.57 \\
SR  & \bf{51.51}     & {39.37} & {46.97}&  {38.61} 
\\
 \hline \hline 
    \end{tabular}
    \end{center}
    \end{table}

We can see from Table \ref{tab_40} that the strategy based on minimizing $\DQex$, with the highest Sharp ratio, is quite competitive compared with other strategies.
In particular, the annualized volatility for minimizing $\DQex$ is relatively lower than that of maximizing the Omega ratio. As mentioned in \cite{SM17}, the Omega ratio is highly sensitive to changes in $t$, and there is no formal guideline for selecting $t$. Moreover, for large values of $t$, maximizing the Omega ratio of a return typically results in a portfolio concentrated in a single asset with the highest mean return over the investment horizon. Consequently, this optimization approach often fails to achieve portfolio diversification when $t$ is large. Therefore, our result suggests that expectile is an ideal choice for the threshold $t$ in our context.
In contrast, the strategies of minimizing ${\rm DQ}^{\rm VaR}_\alpha$ is less favorable with a lower return and higher volatility, which results in lower Sharp ratio. 



Next, we compare the results for $\alpha = 0.01$ and $\alpha = 0.1$, with results presented in Figure \ref{sec_40_2} and Table \ref{tab_40_2}. Generally, increasing values of $\alpha$ improves the portfolio performance. This may benefit from more information on larger data sets due to the increment of $\alpha$, resulting in more robust strategies.
When $\alpha =0.01$, due to the scarcity of tail data, the value of ${\rm DQ}^\VaR_\alpha$ can be minimized to 0 for every period.
 

\begin{figure}[!htbp]
\caption{Wealth processes for  portfolios, 40 stocks, Jan 2016 - Dec 2023.}\label{sec_40_2} 
\centering
           \includegraphics[height=5cm]{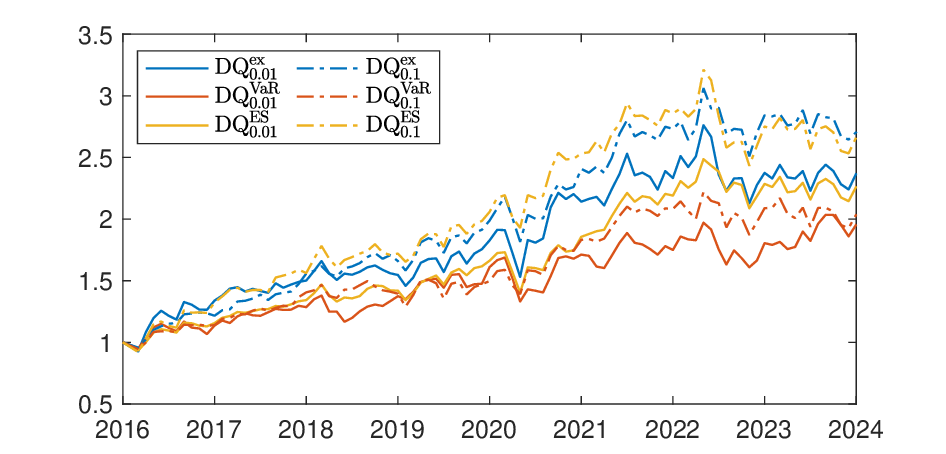}
\end{figure}
  \begin{table}[!htbp]  
\def\arraystretch{1}
  \begin{center} 
    \caption{Summary statistics  for   different portfolio strategies from Jan 2016 to Dec 2023.} \label{tab_40_2} 
  \begin{tabular}{c|ccc|ccc} 
    $\%$ &${\rm DQ}^{\ex}_{0.01}$ &${\rm DQ}^{\VaR}_{0.01}$ &${\rm DQ}^{\ES}_{0.01}$ & ${\rm DQ}^{\ex}_{0.1}$ &${\rm DQ}^{\VaR}_{0.1}$ &${\rm DQ}^{\ES}_{0.1}$  \\    \hline
AR & 11.32    &   8.65  &10.72 & \textbf{13.10} &9.50 &12.84  \\
AV & 19.13   & 16.36 & \textbf{14.69}& 15.12  &15.62 &16.92\\
SR  & {48.06}    & {39.88} &{53.69}&  \bf{72.54}&{47.21} &{59.08}
\\
 \hline \hline 
    \end{tabular}
    \end{center}
    \end{table}

Finally, we conduct the portfolio selection problem on the dataset with the same portfolio in the period 2004-2013 with 2008 financial crisis included.  We use the risk-free rate $r = 4.38\%$ to calculate portfolio statistics.    The results are reported in Figures \ref{sec_40_3}  and Table \ref{tab_40_3} with similar patterns shown in Figure \ref{sec_40} and Table  \ref{tab_40}. We can see that minimizing $\DQex$ is still quite competitive with other strategies.

    \begin{figure}[!htbp]
\caption{Wealth processes for  portfolios, 40 stocks, Jan 2006 - Dec 2013 with $\alpha=0.05$.}\label{sec_40_3} 
\centering
           \includegraphics[height=5cm]{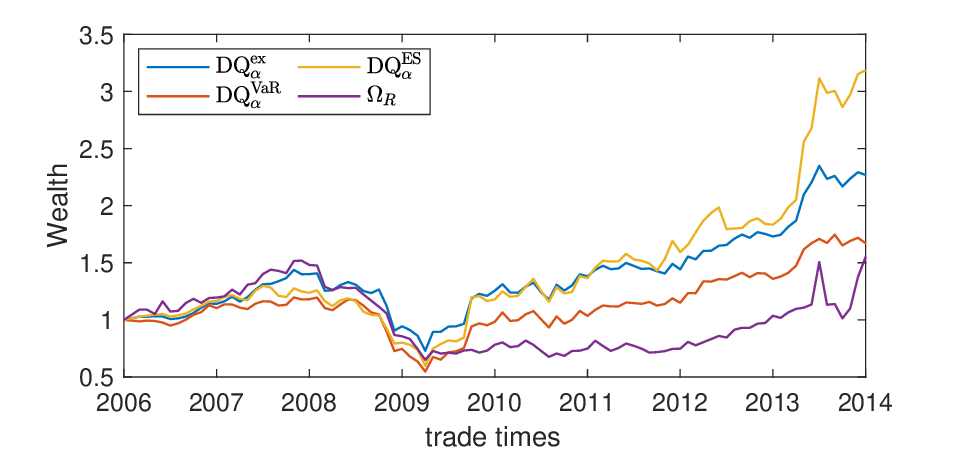}
\end{figure} 
  \begin{table}[!htbp]  
\def\arraystretch{1}
  \begin{center} 
    \caption{Summary statistics  for   different portfolio strategies from Jan 2006 to Dec 2013 with $\alpha=0.05$.} \label{tab_40_3} 
  \begin{tabular}{c|cccc} 
    $\%$ &${\rm DQ}^{\ex}_\alpha$ &${\rm DQ}^{\VaR}_\alpha$ &${\rm DQ}^{\ES}_\alpha$ & $\Omega_R(t_0)$  \\    \hline
AR & 10.46   &   6.47 &\textbf{14.73} & 6.15   \\
AV & \textbf{19.51}  & 20.46 & 24.83 & 24.93 \\
SR  & {31.16}     & {10.19} & \bf{41.74}&  {7.10} \\
 \hline \hline 
    \end{tabular}
    \end{center}
    \end{table}


 \section{Conclusions} \label{sec:7}

In this study, we examined the properties of $\mathrm{DQ}_\alpha^{\ex}$ as an  alternative to $\mathrm{DQ}_\alpha^{\VaR}$ and $\mathrm{DQ}_\alpha^{\ES}$. 
The main message is that $\mathrm{DQ}_\alpha^{\ex}$ enjoys most of the desirable properties of $\mathrm{DQ}_\alpha^{\ES}$ and $\mathrm{DQ}_\alpha^{\VaR}$, and has its own advantages, such as overcoming the small-sample problem and satisfying pseudo-convexity. 
Our findings also showed a strong connection between $\mathrm{DQ}_\alpha^{\ex}$ and the Omega ratio. Moreover,    the portfolio problem of DQ based on expectiles  can be effectively solved using linear programming techniques.

Our general commendation is that, given its various advantages, DQ based on expectiles can be useful in many applications.
It has superior features in terms of balancing gain-loss ratio and small-sample applicability over DQs based on VaR and ES.  
However, we do not intend to suggest that DQ based on expectiles is the only suitable choice, or the best choice, of DQ. In particular,
DQ based on ES shares many properties with DQ based on expectiles, and moreover,
 ES has a special role in some financial applications because ES is the standard risk measure in the current Basel Accords, making DQ based on ES a natural choice of measurement for diversification. 

Future research could explore robust DQ optimization under uncertainty, akin to \cite{PT05}  for the Sharpe ratio and \cite{KZCR14b} for  Omega ratio. This could involve developing techniques that account for potential market fluctuations and model imperfections, ensuring the DQ remains effective in diverse conditions. Additionally, investigating the non-parametric estimation and asymptotic properties of DQ estimators would provide valuable insights into their statistical behavior and reliability. Such studies could enhance our understanding of the DQ's performance across various contexts, leading to more robust applications in portfolio management and risk assessment.

\appendix

\section{Proofs of all results}
\label{app:proof}

\begin{proof}[Proof of Theorem \ref{th:var}]
Since $\ex_\beta$ is coherent for all $\beta \in (0,1)$, we have
\begin{align*}
\mathrm{DQ}^{\ex}_\alpha (\mathbf{X})& = \frac{1}{\alpha} \inf\left\{\beta \in (0,1) :  \ex_{\beta} \left(\sum_{i=1}^n X_i - \sum_{i=1}^n \ex_{\alpha}(X_i)\right) \le0 \right\}\\&= \frac{1}{\alpha} \inf\left\{\beta \in (0,1) :  \ex_{\beta} \left(S-t\right) \le0 \right\}\\
&=\frac{1}{\alpha}\inf \left\{\beta\in (0,1): S-t \in \mathcal{A}_{\ex_{\beta}}\right\}=\frac{1}{\alpha} \times \frac{\mathbb{E}\left[ (S-t) _+\right]}{\mathbb{E}\left[\left|S-t\right|\right]}.
\end{align*} 
The last equality follows from the acceptance set of $\ex_\beta$ in \eqref{eq:acceptance}.
Furthermore, for any $X\in L^1$, we have
$$
\E[(X-y)_{-}]=\int_{-\infty}^y(y-x) \mathrm{d} F_X(x)=y F_X(y)-\int_{-\infty}^y x \mathrm{d} F_X(x),
$$ 
and
$$
\begin{aligned}
\mathbb E[|X-y|] & =\int_{-\infty}^y(y-x) \mathrm{d} F_X(x)+\int_y^{\infty}(x-y) \mathrm{d} F_X(x) \\
& =2\left(y F_X(y)-\int_{-\infty}^y x \mathrm{~d} F_X(x)\right)+\E[X]-y.
\end{aligned}
$$
Therefore, 
\begin{align*}
\mathrm{DQ}^{\ex}_\alpha (\mathbf{X})=\frac{1}{\alpha}
\left (1- \frac{\mathbb{E}\left[(S-t)_-\right]}{\mathbb{E}\left[\left|S-t\right|\right]}\right)
=\frac 1 \alpha \left(1-\widetilde{F}_{S}\left(t\right )\right),
\end{align*} 
where $\widetilde F_{S}$ is defined in \eqref{t_F}. 
\end{proof}
\begin{proof}[Proof of Proposition \ref{th:ex-01}]
   (i) 
   Since $\ex_\alpha(\sum_{i=1}^n X_i) \le \sum_{i=1}^n \ex_\alpha(X_i)$, it is clear that $\DQex(\mathbf X)$ takes value in $[0,1]$ for any $\mathbf X \in \X$. Next, we show that every point in $[0,1]$ is attainable by $\DQex$.
   
   Let $X_1$ and $X_2$ follow a uniform distribution on $[-1,1]$. For 
 any $t\in [0,2]$, Theorem 3.1 of \cite{WW16} shows that we can find $(X_1, X_2)$ such that $X_1+X_2 \sim \mathrm{U}[-t,t]$. Let $\mathbf X=(X_1, X_2,0, \dots, 0)$. We can compute
$$
y_1=\ex_\alpha(X_1)=\ex_\alpha(X_2)=\frac{1-2\sqrt{\alpha(1-\alpha)}}{1-2\alpha}\in(0,1]$$  for $ \alpha\in[0,1/2)$.  
 Meanwhile, we have $$\widetilde{F}_{S}(y)=
\widetilde{F}_{X_1+X_2}(y)=\frac{(y+t)^2}{2(y^2+t^2)}.
$$
Thus,  by \eqref{eq:DQ_ex},  we have 
$$\begin{aligned}\mathrm{DQ}^{\ex}_\alpha (\mathbf{X})&=\frac{1-\widetilde{F}_{S}( 2y_1)}{\alpha}=\frac{(2y_1-t)^2}{2\alpha(4y_1^2+t^2)}.\end{aligned}$$
It is easy to check that if $t=2y_1$, then   $\mathrm{DQ}^\ex_\alpha(\mathbf X)=0$; whereas if  $t=2y^2_1$,  then $\mathrm{DQ}^\ex_\alpha(\mathbf X)=1$. Since $\mathrm{DQ}^{\ex}_\alpha (\mathbf{X})$ is continuous in $t\in[0,2]$, it follows  that every point in $[0,1]$ is attained by $\mathrm{DQ}^\ex_\alpha$.

(ii)  The first part is straightforward from the alternative formula  \eqref{eq:var-alter}. In particular, if $\sum_{i=1}^n X_i$ is a constant, we have $\ex_{0}\left(\sum_{i=1}^n X_i\right)=\ex_{\alpha}\left(\sum_{i=1}^n X_i\right)\leq\sum_{i=1}^n\ex_{\alpha}( X_i)$, which gives $\mathrm{DQ}_\alpha^\ex(\mathbf{X})=0$.  

(iii) Note that $\sum_{i=1}^n \lambda_i>0$ and $\E\left[\left(X-\ex_\alpha(X)\right)_+\right]>0$. Therefore,
\begin{align*}
\mathrm{DQ}_\alpha^\ex(\lambda_1X, \dots, \lambda_nX)&=\frac{1}{\alpha}
\left\{\mathbb{E}\left[\left(\sum_{i=1}^n \lambda_i \left(X-\ex_\alpha(X)\right)\right)_+\right]\Big/\mathbb{E}\left[\left|\sum_{i=1}^n \lambda_i \left(X-\ex_\alpha(X)\right)\right|\right]\right\}\\
&=\frac{1}{\alpha}
\left\{\E\left[\left(X-\ex_\alpha(X)\right)_+\right]\Big/\mathbb{E}\left[\left| \left(X-\ex_\alpha(X)\right)\right|\right]\right\}\\
&=\frac{1}{\alpha}\frac{1}{1+\frac{1-\alpha}{\alpha}}=1. 
\end{align*}
where the last equality follows from  \eqref{eq:ex}.
\end{proof}

 \begin{proof}[Proof of Proposition \ref{prop:2}]
(i) Since $\DQex(\mathbf X)>0$,  by (ii) in Proposition \ref{th:ex-01}, we have
$\p(\sum_{i=1}^n X_i\le \sum_{i=1}^n \ex_\alpha(X_i))<1$. Hence, we have 
 $$\ex_0\left(\sum_{i=1}^n X_i\right)=\esssup\left(\sum_{i=1}^n X_i\right)> \sum_{i=1}^n \ex_\alpha(X_i).$$ 
Moreover, $\sum_{i=1}^n X_i$ is not a constant. Hence,
$\beta \mapsto \ex_\beta(\sum_{i=1}^n X_i)$ is continuous and strictly decreasing (Theorem 1 of \cite{NP87}). As $\alpha \in (0,1/2)$, we have $ \ex_{\alpha} \left(\sum_{i=1}^n X_i\right) \leq  \sum_{i=1}^n \ex_\alpha(X_i) $. Hence,  there exists unique $c \in (0, 1]$ such that 
$\ex_{c\alpha} \left(\sum_{i=1}^n X_i\right) = \sum_{i=1}^n \ex_\alpha(X_i)$
holds. By Definition \ref{de:DQex}, we have $\DQex(\mathbf X)=\alpha^*/\alpha=c$.

(ii) First, we assume $\DQex(\mathbf X)>0$. Since   $\ex_{1-\alpha}(-X)=-\ex_{\alpha}(X)$  for any  $X\in L^1$,   we have  \begin{align*}
(1-\alpha)\mathrm{DQ}^{\ex}_{1-\alpha} (-\mathbf{X})
& =  \inf\left\{\beta \in (0,1) : \ex_{\beta} \left(-\sum_{i=1}^n X_i\right) \le  \sum_{i=1}^n \ex_{1-\alpha}(-X_i) \right\}\\
&=  \inf\left\{\beta \in (0,1) :  -\ex_{1-\beta} \left(\sum_{i=1}^n X_i\right) \le  -\sum_{i=1}^n {\ex}_\alpha(X_i) \right\}\\&= 
 1-\sup\left\{\gamma \in (0,1) :  \ex_{\gamma} \left(\sum_{i=1}^n X_i\right) \ge  \sum_{i=1}^n {\ex}_\alpha(X_i)\right\}.
\end{align*}   
By (i), we have that $\DQex(\mathbf X)$ 
is the unique solution of $c$ that satisfies $\ex_{c\alpha}(\sum_{i=1}^n X_i)=\sum_{i=1}^n \ex_\alpha(X_i)$. Hence, 
$$\sup\left\{\gamma \in (0,1) :  \ex_{\gamma} \left(\sum_{i=1}^n X_i\right) \ge  \sum_{i=1}^n \ex_\alpha(X_i) \right\}=\alpha\DQex(\mathbf X).$$
Therefore,
$(1-\alpha)\mathrm{DQ}^{\ex}_{1-\alpha} (-\mathbf{X})=1-\alpha\DQex(\mathbf X)$,
which gives \eqref{eq:sym}.	

If $\DQex(\mathbf X)=0$, then we have $\sum_{i=1}^n X_i \le \sum_{i=1}^n \ex_\alpha(X_i)$ a.s.~from Proposition \ref{th:ex-01} (ii). Therefore,
$\sum_{i=1}^n -X_i \ge \sum_{i=1}^n -\ex_\alpha(X_i)=\sum_{i=1}^n \ex_{1-\alpha}(X_i)$ a.s.
Since $\mathbf X$ is non-degenerate, $\ex_\beta(\sum_{i=1}^n X_i)$ is strictly decreasing with $\beta$. Therefore, $$\ex_\beta\left(-\sum_{i=1}^n X_i\right)>\ex_{1}\left(-\sum_{i=1}^n X_i\right)\ge \sum_{i=1}^n \ex_{1-\alpha}(X_i) ~\mbox{for all}~ \beta\in (0,1),$$ which implies $\DQex(-\mathbf X)=1/(1-\alpha)$ from Definition \ref{de:DQex}.  It is easy to check  $\alpha \DQex(\mathbf X)+(1-\alpha)\DQES(-\mathbf X)=1$.
\end{proof}

\begin{proof}[Proof of Proposition \ref{uncorrelated}]
Since expectile is a coherent risk measure, the proof $\lim_{n \to \infty} \DQex(\mathbf X)=0$ follows the same structure as the proof for  $\lim_{n \to \infty} \DQES(\mathbf X)=0$ in \cite[ Theorem 2]{HLW23}. For clarity, we outline our proof below.

Since $\mathrm{DQ}^\ex_\alpha$ is location invariant, we  can assume that $\E[X_i]=0$ for $i=1, 2, \dots$. 
Let $\mathbf X_n=(X_1, \dots, X_n)$ and $S_n=\sum_{i=1}^n X_i$. By the $L^2$-Law of Large Numbers in \citet[Theorem 2.2.3]{D19}, we have $S_n/n \buildrel L^2 \over \rightarrow 0$.
By \citet[Corollary 7.10]{R13}),   expectile is $L^1$-continuous.
Hence, $\ex_\beta(S_n/n) \to 0$  as $n\to\infty$ for all $\beta \in (0,1)$.
Let $\epsilon=\inf_{i\in \mathbb{N}} \ex_\alpha(X_i)$.
As a result, for every $\beta \in (0,1)$, there exists $N_\beta$ such that $\ex_\beta(S_n/n)<\epsilon$ for all $n>N_\beta$. Therefore, we have  
$$\alpha^*=\inf\left\{\beta \in (0,1):\ex_\beta(S_n)\le \sum_{i=1}^n \ex_\alpha(X_i)\right\}\le \inf\left\{\beta \in (0,1):\ex_\beta(S_n/n)\le \epsilon\right\} \to 0$$
as $n \to \infty$.
Hence, we have $\mathrm{DQ}^\ex_\alpha(\mathbf X_n)=\alpha^*/\alpha \to 0$ as $n \to \infty$.
 \end{proof}

\begin{proof}[Proof of Proposition \ref{quasi-convex}]
First, consider the case   $\mathbf w=\mathbf 0$. For $\mathbf v\in \R_+^n$,  using scale invariance,  we have
\begin{align*}
 \mathrm{DQ}_\alpha^{\ex}((\lambda\mathbf 0+(1-\lambda)\mathbf v )\odot \mathbf X)=\mathrm{DQ}_\alpha^{\ex}((1-\lambda)\mathbf v \odot \mathbf X)=\mathrm{DQ}_\alpha^{\ex}(\mathbf v \odot \mathbf X).
\end{align*}
Since $\DQex(\mathbf 0 \odot \mathbf X)=0$ and $\DQex$ is nonnegative, we have 
$\max\{\DQex(\mathbf 0 \odot \mathbf X), \DQex(\mathbf v \odot \mathbf X)\}=\DQex(\mathbf v \odot \mathbf X)$.

Next, let $\mathbf w=(w_1, \dots, w_n), \mathbf v=(v_1, \dots, v_n) \in \R_+^n\setminus\{\mathbf 0\}$. Define
$w=\sum_{i=1}^n w_i>0$, $v=\sum_{i=1}^n v_i>0$, $\mathbf w'=\mathbf w/w$, and $\mathbf v'=\mathbf v/v$. It is clear that  $\mathbf w'\in \Delta_n$ and $\mathbf v'\in \Delta_n$.
For any $\lambda \in [0,1]$, by   scale invariance (SI) of $\DQex$ and its quasi-convexity  (QC) on $\Delta_n$, we have
\begin{align*}
 &\mathrm{DQ}_\alpha^{\ex}((\lambda\mathbf w+(1-\lambda)\mathbf v )\odot \mathbf X)\\
 &= \mathrm{DQ}_\alpha^{\ex}\left((\lambda w+(1-\lambda)v)\left(\frac{\lambda w}{\lambda w+(1-\lambda)v}\mathbf w'+\frac{(1-\lambda)v}{\lambda w+(1-\lambda)v}\mathbf v' \right)\odot \mathbf X\right)\\
\mbox{\footnotesize{(by SI)}~~~} &= \mathrm{DQ}_\alpha^{\ex}\left(\left(\frac{\lambda w}{\lambda w+(1-\lambda)v}\mathbf w'+\frac{(1-\lambda)v}{\lambda w+(1-\lambda)v}\mathbf v' \right)\odot \mathbf X\right)\\
 \mbox{\footnotesize{(by QC)}~~~}& \le \max\left\{\mathrm{DQ}_\alpha^{\ex}(\mathbf w' \odot \mathbf X), \mathrm{DQ}_\alpha^{\ex}(\mathbf v' \odot \mathbf X) \right\}\\
\mbox{\footnotesize{(by SI)}~~~} & = \max\left\{\mathrm{DQ}_\alpha^{\ex}(w\mathbf w' \odot \mathbf X), \mathrm{DQ}_\alpha^{\ex}(v\mathbf v' \odot \mathbf X) \right\} \\&= \max\left\{\mathrm{DQ}_\alpha^{\ex}(\mathbf w \odot \mathbf X), \mathrm{DQ}_\alpha^{\ex}(\mathbf v \odot \mathbf X) \right\}.
\end{align*}
Therefore,  the mapping $\mathbf w \mapsto \mathrm{DQ}_\alpha^{\ex}$ is quasi-convex on $\R^n_+$.  
\end{proof}
\begin{proof}[Proof of Theorem \ref{prop:pseudo}]
We assume $\mathbf X$ is not a constant vector; otherwise $f(\mathbf w)=0$ and there is nothing to show. 
By Theorem \ref{th:var}, we have   
\begin{align}\label{eq:f}
f(\mathbf w )&= \frac{\mathbb{E}[(\mathbf  w^\top(\mathbf X- \mathbf{x}_{\alpha}^{\ex} ))_+]}{\alpha(2\mathbb{E}[(\mathbf  w^\top(\mathbf X- \mathbf{x}_{\alpha}^{\ex} ))_+]-\mathbb{E}[\mathbf  w^\top(\mathbf X- \mathbf{x}_{\alpha}^{\ex} )])},
\end{align} 
and note that the denominator in the right-hand side of \eqref{eq:f} is positive.
We aim to show for all $\mathbf w=(w_1, \dots, w_n), \mathbf v=(v_1, \dots, v_n) \in (0, +\infty)^n$, if  $f(\mathbf v) <f(\mathbf w )$, then   $\nabla f(\mathbf w) \cdot(\mathbf v-\mathbf w) < 0$. 
We can compute  
$$\begin{aligned}
&\frac{\partial f(\mathbf w )}{\partial w_{i}}
\\&= \frac{\E[  (X_i-\ex_\alpha(X_i))\id_{\{\mathbf w^\top (\mathbf X-\mathbf x^{\ex}_\alpha)>0\}}](2\mathbb{E}[(\mathbf  w^\top(\mathbf X- \mathbf{x}_{\alpha}^{\ex} ))_+]-\mathbb{E}[\mathbf  w^\top(\mathbf X- \mathbf{x}_{\alpha}^{\ex} )])}{ \alpha(2\mathbb{E}[(\mathbf  w^\top(\mathbf X- \mathbf{x}_{\alpha}^{\ex} ))_+]-\mathbb{E}[\mathbf  w^\top(\mathbf X- \mathbf{x}_{\alpha}^{\ex} )])^2}
\\&~~~~-\frac{\mathbb{E}[(\mathbf  w^\top(\mathbf X- \mathbf{x}_{\alpha}^{\ex} ))_+] (2  \E[ (X_i-\ex_\alpha(X_i))\id_{\{\mathbf w^\top (\mathbf X-\mathbf x^{\ex}_\alpha)>0\}}] -\mathbb{E}[(X_i- \mathbf{x}_{\alpha}^{\ex}(X_i) )])}{ \alpha (2\mathbb{E}[(\mathbf  w^\top(\mathbf X- \mathbf{x}_{\alpha}^{\ex} ))_+]-\mathbb{E}[\mathbf  w^\top(\mathbf X- \mathbf{x}_{\alpha}^{\ex} )])^2  }\\&=\frac{\mathbb{E}[(\mathbf  w^\top(\mathbf X- \mathbf{x}_{\alpha}^{\ex} ))_+] \mathbb{E}[\mathbf  ( X_i- \mathbf{x}_{\alpha}^{\ex}(X_i) )]  -\E[ (X_i-\ex_\alpha(X_i))\id_{\{\mathbf w^\top (\mathbf X-\mathbf x^{\ex}_\alpha)>0\}}] \mathbb{E}[\mathbf  w^\top(\mathbf X- \mathbf{x}_{\alpha}^{\ex} )]}{ \alpha (2\mathbb{E}[(\mathbf  w^\top(\mathbf X- \mathbf{x}_{\alpha}^{\ex} ))_+]-\mathbb{E}[\mathbf  w^\top(\mathbf X- \mathbf{x}_{\alpha}^{\ex} )])^2 }.
\end{aligned}$$ 
It is easy to check that
$\nabla f(\mathbf w) \cdot\mathbf w=0$.
If   $f(\mathbf v) < f(\mathbf w)$,  by \eqref{eq:f}, we have  
\begin{equation}
    \label{eq:str-1}\E[(\mathbf  v^\top(\mathbf X- \mathbf{x}_{\alpha}^{\ex} ))]\mathbb{E}[(\mathbf  w^\top(\mathbf X- \mathbf{x}_{\alpha}^{\ex} )_+)]- \mathbb{E}[\mathbf  w^\top(\mathbf X- \mathbf{x}_{\alpha}^{\ex} )] \E[(\mathbf  v^\top(\mathbf X- \mathbf{x}_{\alpha}^{\ex} ))_+]< 0.
\end{equation}   
Since $\ex_\alpha(X_i)>\E[X_i]$ for $\alpha \in (0,1/2)$, we have $\E\left[\mathbf  w^\top(\mathbf X- \mathbf{x}_{\alpha}^{\ex} )\right]<0$. Moreover,
it is clear that 
\begin{equation}
    \label{eq:str-2}\E[(\mathbf v^\top (\mathbf X-\mathbf x_\alpha^\ex))_+]\ge\E[(\mathbf v^\top (\mathbf X-\mathbf x_\alpha^\ex))\id_{\{\mathbf w^\top(\mathbf X-\mathbf x_{\alpha}^\ex)>0\}}].
    \end{equation}
Hence, using \eqref{eq:str-1} and \eqref{eq:str-2},
\begin{align*}
&\nabla f(\mathbf w) \cdot(\mathbf v-\mathbf w)\\
&= \nabla f(\mathbf w) \mathbf v - 0\\
&= 
\frac{ \mathbb{E}[\mathbf  v^\top(\mathbf X- \mathbf{x}_{\alpha}^{\ex} )] \mathbb{E}[(\mathbf  w^\top(\mathbf X- \mathbf{x}_{\alpha}^{\ex} ))_+] - \mathbb{E}[\mathbf  w^\top(\mathbf X- \mathbf{x}_{\alpha}^{\ex} )] \E[\mathbf  v^\top(\mathbf X- \mathbf{x}_{\alpha}^{\ex} )\id_{\{\mathbf w^\top (\mathbf X-\mathbf x^{\ex}_\alpha)>0\}}]}{\alpha (2\mathbb{E}[(\mathbf  w^\top(\mathbf X- \mathbf{x}_{\alpha}^{\ex} ))_+]-\mathbb{E}[\mathbf  w^\top(\mathbf X- \mathbf{x}_{\alpha}^{\ex} )])^2 }\\
&\le 
\frac{ \mathbb{E}[\mathbf  v^\top(\mathbf X- \mathbf{x}_{\alpha}^{\ex} )] \mathbb{E}[(\mathbf  w^\top(\mathbf X- \mathbf{x}_{\alpha}^{\ex} ))_+] - \mathbb{E}[\mathbf  w^\top(\mathbf X- \mathbf{x}_{\alpha}^{\ex} )] \E[\mathbf  v^\top(\mathbf X- \mathbf{x}_{\alpha}^{\ex} )_+]}{\alpha (2\mathbb{E}[(\mathbf  w^\top(\mathbf X- \mathbf{x}_{\alpha}^{\ex} ))_+]-\mathbb{E}[\mathbf  w^\top(\mathbf X- \mathbf{x}_{\alpha}^{\ex} )])^2 }<0
\end{align*}
which implies that   $f$ is  pseudo-convex on $(0, \infty)^n$.  
\end{proof}
\begin{proof}[Proof of Proposition \ref{thm:2}] 
By Theorem \ref{th:var}, we have 
\begin{align*} 
{\rm DQ}^{\ex}_{\alpha}(\mathbf w \odot \mathbf X)
=\frac{1}{\alpha}\frac{\mathbb{E}[(\mathbf  w^\top(\mathbf X- \mathbf{x}_{\alpha}^{\ex} ))_+]}{2\mathbb{E}\left[\left(\mathbf  w^\top(\mathbf X- \mathbf{x}_{\alpha}^{\ex} )\right)_+\right]-\mathbb{E}\left[\mathbf  w^\top(\mathbf X- \mathbf{x}_{\alpha}^{\ex} )\right]}.\end{align*} 
Since $\E[X_i-\ex_\alpha(X_i)]<0$ for   $i\in [n]$ and $\alpha\in(0,1/2)$,  we have $\mathbb{E}\left[\mathbf  w^\top(\mathbf{x}_{\alpha}^{\ex}-\mathbf X )\right] > 0$ for any  $\mathbf{w} \in \Delta_n$. 
If there exists $\mathbf  w_0\in\Delta_n$ such that $\mathbb{E}\left[\left(\mathbf  w_0^\top(\mathbf X- \mathbf{x}_{\alpha}^{\ex} )\right)_+\right]=0$, then $$\min _{\mathbf{w} \in \Delta_n}  {\rm DQ}^{\ex}_{\alpha}(\mathbf w \odot \mathbf X)=\min_{\mathbf w \in \Delta_n}~\frac{\mathbb{E}[(\mathbf  w^\top(\mathbf X- \mathbf{x}_{\alpha}^{\ex} ))_+]}{\mathbb{E}\left[ \mathbf  w^\top (\mathbf{x}_{\alpha}^{\ex} - \mathbf X)\right]}= 0.$$   Otherwise, we have 
$$ {\rm DQ}^{\ex}_{\alpha}(\mathbf w \odot \mathbf X) =\frac{1}{\alpha}\frac{1}{2-\frac{\mathbb{E}\left[\mathbf  w^\top(\mathbf X- \mathbf{x}_{\alpha}^{\ex} )\right]}{\mathbb{E}\left[\left(\mathbf  w^\top(\mathbf X- \mathbf{x}_{\alpha}^{\ex} )\right)_+\right]}}.$$ 
Since $\mathbb{E}\left[\mathbf  w^\top(\mathbf{x}_{\alpha}^{\ex}-\mathbf X )\right] > 0$ for any  $\mathbf{w} \in \Delta_n$, minimizing $\mathrm{DQ}^{\ex}_\alpha (\mathbf w \odot\mathbf{X})$ is equivalent to solving  \eqref{optimal}. 
Let $m>0$. It is clear that  minimizing \eqref{optimal}, subject to  $\mathbb{E}\left[\mathbf  w^\top(\mathbf X- \mathbf{x}_{\alpha}^{\ex} )\right]=m$, is equivalent to  minimize $\mathbb{E}\left[\left(\mathbf  w^\top(\mathbf X- \mathbf{x}_{\alpha}^{\ex} )\right)_+\right]/m$ over $\mathbf{w} \in \Delta_n$, subject to  $\mathbb{E}\left[\left(\mathbf  w^\top(\mathbf X- \mathbf{x}_{\alpha}^{\ex} )\right)\right]=m$. Hence, by varing $m \in \R_+$, we get \eqref{eq:opt3}.
 \end{proof}

\begin{proof}[Proof of Proposition \ref{prop:ellip}] Since  $\mathbf{X} \sim \mathrm{E}_{n}(\boldsymbol{\mu}, \Sigma, \tau)$,  the linear  structure of ellipitical distributions gives    $\sum_{i=1}^n X_i \sim  \mathrm{E}_{1}( \mathbf 1^\top \boldsymbol{\mu}, \mathbf 1^\top\Sigma \mathbf 1, \tau)$.  That is, $\sum_{i=1}^n X_i \laweq \sum_{i=1}^n \mu_i+\Vert \mathbf{1}^\top A\Vert_2 Y$, where $A$ is the Cholesky decomposition of $\Sigma$, and  $\Vert \cdot \Vert_2$ denotes the Euclidean norm.
		Also,  we have $ \ex_\alpha(X_i)=\mu_i+\Vert  \mathbf{e}^\top_iA\Vert_2 \ex_\alpha (Y).$
Note that 
 \begin{align*}
\mathrm{DQ}^{\ex}_\alpha (\mathbf{X})& = \frac{1}{\alpha} \inf\left\{\beta \in (0,1) :  \ex_{\beta} \left( \Vert \mathbf{1}^\top A\Vert_2Y- \sum_{i=1}^n\Vert  \mathbf{e}^\top_iA\Vert_2  \ex_\alpha(Y)  \right) \le0 \right\}\\&
=\frac{1}{\alpha}\inf \left\{\beta\in (0,1): Y-k_\Sigma \ex_\alpha(Y)    \in \mathcal{A}_{\ex_{\beta}}\right\}. 
\end{align*} 
Following a similar approach as in the proof of Theorem \ref{th:var}, we obtain \eqref{DQ_elli}. Since $\ex_{\alpha}(Y)>0$  for $\alpha \in (0,1/2)$, 			 $ {\rm DQ}_\alpha^{\ex}(\mathbf X)$ is decreasing  in $k_\Sigma$ for $\alpha \in (0,1/2)$. In particular, if $k_\Sigma=1$, together with  \eqref{trans_p}, we have $ {\rm DQ}_\alpha^{\ES}(\mathbf X)=1$. 
\end{proof}	

\begin{proof}[Proof of Proposition \ref{prop:7}] Since  $\mathbf{X} \sim \mathrm{E}_{n}(\boldsymbol{\mu}, \Sigma, \tau)$,  the linear  structure of ellipitical distributions gives    $\sum_{i=1}^n w_i X_i \sim  \mathrm{E}_{1}( \mathbf w^\top \boldsymbol{\mu}, \mathbf w^\top\Sigma \mathbf w, \tau)$.  That is, $\sum_{i=1}^n w_iX_i \laweq \sum_{i=1}^n w_i\mu_i+\Vert \mathbf{w}^\top A\Vert_2 Y$, where $A$ is the Cholesky decomposition of $\Sigma$. Moreover,  $ \ex_\alpha(w_iX_i)=w_i\mu_i+\Vert w_i \mathbf{e}^\top_iA\Vert_2 \ex_\alpha (Y).$ Hence, we have
	\begin{align*}
\mathrm{DQ}^{\ex}_\alpha (\mathbf w \odot \mathbf X)& =\frac{\mathbb{E}\left[\left(Y-k_{\mathbf w \Sigma \mathbf w^\top}\ex_\alpha(Y)\right)_+\right]}{\alpha \mathbb{E}\left[\left|Y-k_{\mathbf w \Sigma \mathbf w^\top}\ex_\alpha(Y)\right|\right]}=\frac{1}{\alpha\left (1+\frac{ \mathbb{E}\left[\left(Y-k_{\mathbf w \Sigma \mathbf w^\top}\ex_\alpha(Y)\right)_-\right]}{\mathbb{E}[(Y-k_{\mathbf w \Sigma \mathbf w^\top}\ex_\alpha(Y))_+]}\right)}.
\end{align*} 	
Since $\ex_\alpha(Y)>0$, it is clear that $k_{\mathbf w \Sigma \mathbf w^\top} \mapsto {\rm DQ}_\alpha^{\ex}(\mathbf w \odot \mathbf X)$ is decreasing for $\alpha \in(0,1/2)$, thus minimizing  $\mathrm{DQ}^{\ex}_\alpha (\mathbf w \odot \mathbf X)$ is equivalent to maximizing $k_{\mathbf w \Sigma \mathbf w^\top}$.
\end{proof}		

\begin{proof}[Proof of Proposition \ref{prop:lim}]	If  $\mathbf{X} \in \mathrm{MRV}_{\gamma}(\Psi)$ with $\gamma>0$, we have (Lemma  2.2 of \cite{ME13})
		$$
		\lim _{\alpha \downarrow 0}\frac{ \VaR_{\alpha} \left(\sum_{i=1}^n X_i\right) }{ \VaR_{\alpha} \left(\|\mathbf{X}\|\right)} =  \eta_{\mathbf{1}}^{1 / \gamma},
		$$
		and
		$$
		\lim _{\alpha \downarrow 0}\sum_{i=1}^n \frac{ \VaR_{\alpha} \left(X_i\right) }{ \VaR_{\alpha} \left(\|\mathbf{X}\|\right)} = \sum_{i=1}^n\eta_{\mathbf{e}_i}^{1 / \gamma}.
		$$	
Since $\ex_\alpha$ is continuous, strictly decreasing in $\alpha$,   we have $\ex_{\alpha^*}(\sum_{i=1}^n X_i)=\sum_{i=1}^n \ex_{\alpha} (X_i)$.  Additionally,  as $\sum_{i=1}^n X_i\in \mathrm{RV}_\gamma$ and $\ex_\alpha(\sum_{i=1}^n X_i)/\VaR_\alpha (\sum_{i=1}^n X_i)\to (\gamma-1)^{-1/\gamma}$ as $\alpha \downarrow 0$ for $X\in \mathrm{RV}_\gamma$ with $\gamma>1$,  it follows that
$$\begin{aligned}
		&\lim _{\alpha \downarrow 0}\frac{ \ex_{\alpha} \left(\sum_{i=1}^n X_i\right) }{\ex_{\alpha^*}(\sum_{i=1}^n X_i)}  = \lim _{\alpha \downarrow 0}\frac{ \ex_{\alpha} (\sum_{i=1}^n  X_i )}{\sum_{i=1}^n   \ex_{\alpha} ( X_i ) }  = \lim _{\alpha \downarrow 0} \frac{ \VaR_{\alpha} (\sum_{i=1}^n  X_i )}{\sum_{i=1}^n   \VaR_{\alpha} ( X_i ) }=  \frac{\eta_{\mathbf{1}}^{1 / \gamma}}{\sum_{i=1}^{n} \eta_{\mathbf{e}_{i}}^{1 / \gamma}}.
			\end{aligned}$$ 
	Let $c=\alpha^*/\alpha$.  Since $\sum_{i=1}^n X_i\in {\rm RV}_\gamma$, for $c>0$, we have 
		$$\frac{ \ex_{\alpha} \left(\sum_{i=1}^n X_i\right) }{\ex_{c\alpha} \left(\sum_{i=1}^n  X_i\right) }= \frac{ \VaR_{\alpha} \left(\sum_{i=1}^n X_i\right) }{\VaR_{c\alpha} \left(\sum_{i=1}^n  X_i\right) }\simeq \left(\frac{1 }{c
		}\right)^{-1/\gamma}
		\mbox{~~~~as }  \alpha \downarrow 0.$$
	Hence,
		$$
		{\rm DQ}^{\ex}_{\alpha}(\mathbf w \odot \mathbf X) = \frac{\alpha^*}{\alpha}=c\to \frac{\eta_{\mathbf{1}}}{\left(\sum_{i=1}^{n}  \eta_{\mathbf{e}_{i}}^{1 / \gamma}\right)^\gamma}.
		$$

Furthermore, if $X_1, \dots, X_n$ are  nonnegative random variables, by Example 3.1 of \cite{ELW09}, we have 
 $$\left(\frac{\eta_{\mathbf 1}}{\eta_{\mathbf e_i}}\right)^{1/\gamma}=  \lim_{\alpha\downarrow 0}\frac{\ex_\alpha\left(\sum_{i=1}^n X_i\right)}{\ex_\alpha(X_1)}=\lim_{\alpha\downarrow 0}\frac{\VaR_\alpha\left(\sum_{i=1}^n X_i\right)}{\VaR_\alpha(X_1)}=n^{1/\gamma},~~i\in [n]$$ which implies that $\eta_{\mathbf 1}/\eta_{\mathbf e_i} =n$ for $i\in [n]$.  Hence, $\lim _{\alpha \downarrow 0} {\rm DQ}^{\ex}_{\alpha}(\mathbf X)=n^{1-\gamma}$. Moreover, if $\gamma\downarrow 1,$ then $ {\rm DQ}^{\ex}_{\alpha}(\mathbf X)\to 1$.
\end{proof}

\end{document}